\documentclass[letter,11pt]{article}
\usepackage{fullpage}
\usepackage{graphicx}
\graphicspath{{./figs/}}
\usepackage{amsmath,amssymb,amsthm}

\makeatletter
\newcommand{\leqnomode}{\tagsleft@true\let\veqno\@@leqno}
\newcommand{\reqnomode}{\tagsleft@false\let\veqno\@@eqno}
\makeatother

\newtheorem{corollary}{Corollary}
\newtheorem{lemma}{Lemma}
\newtheorem{definition}{Definition}
\newtheorem{theorem}{Theorem}

\newtheorem{claim}{Claim}
      \numberwithin{claim}{lemma}

\newtheorem{proposition}{Proposition}
\newtheorem{remark}{Remark}
\newtheorem{observation}{Observation}

\usepackage{wrapfig}

\usepackage{comment}

\usepackage[algosection,ruled,lined,linesnumbered,longend]{algorithm2e}

\usepackage{algcompatible}
\usepackage{color}
\usepackage{bbm}
\usepackage{mathbbol}
\usepackage{tabularx}
\usepackage{enumerate}

\usepackage{thmtools}
\usepackage{xcolor}
\colorlet{agcolor}{blue!70!white}
\colorlet{jkcolor}{green!60!black}

\usepackage[textwidth=20mm]{todonotes}

\newcommand{\0}{\mathbb{0}} 
 
\newcommand{\C}{\ensuremath{\mathcal{C}}}

\newcommand{\calT}{\mathcal{T}}
\newcommand{\tilT}{\mathcal{T}}

\usepackage[utf8]{inputenc}

\newcommand{\ECT}{\ensuremath{\textsc{ECT}}}

\usepackage[numbers,sort&compress]{natbib}
\usepackage[bookmarks=true,colorlinks=true,citecolor=blue,linkcolor=red]{hyperref}

\begin{document}

\title{Hitting Weighted Even Cycles in Planar Graphs}
\author{Alexander G{\"o}ke\thanks{Hamburg University of Technology, Institute for Algorithms and Complexity, Hamburg, Germany. \texttt{alexander.goeke@tuhh.de}}
  \and Jochen Koenemann\thanks{University of Waterloo, Waterloo, Canada. \texttt{jochen@uwaterloo.ca}}
  \and Matthias Mnich\thanks{Hamburg University of Technology, Institute for Algorithms and Complexity, Hamburg, Germany. \texttt{matthias.mnich@tuhh.de}}
  \and Hao Sun\thanks{University of Waterloo, Waterloo, Canada. \texttt{hao@uwaterloo.ca}}}
\date{}
\maketitle              
\begin{abstract}
  A classical branch of graph algorithms is graph transversals, where one seeks <a minimum-weight subset of nodes in a node-weighted graph $G$ which
  intersects all copies of subgraphs~$F$ from a fixed family $\mathcal F$.
  Many such graph transversal problems have been shown to admit polynomial-time approximation schemes (PTAS) for \emph{planar} input graphs $G$, using a variety of techniques like the shifting technique (Baker, J. ACM 1994), bidimensionality (Fomin et al., SODA 2011), or connectivity domination (Cohen-Addad et al., STOC 2016).
  These techniques do not seem to apply to graph transversals with parity constraints, which have recently received significant attention, but for which no PTASs are known.

  In the {\em even-cycle transversal} (\ECT) problem, the goal is to find a minimum-weight hitting set for the set of even cycles in an undirected graph.
  For ECT, Fiorini et al. (IPCO 2010) showed that the integrality gap of the standard covering LP relaxation is $\Theta(\log n)$, and that
  adding sparsity inequalities reduces the integrality gap to~10.

  \quad Our main result is a primal-dual algorithm that yields a $47/7\approx6.71$-approximation for ECT on node-weighted planar graphs, and an integrality gap of the same value for the standard LP relaxation on node-weighted planar graphs.

\end{abstract}
\section{Introduction}
\label{sec:introduction}
Transversal problems in graphs have received a significant amount of attention from the perspective of algorithm design.
Such problems take as input a node-weighted graph~$G$, and seek a minimum-weight subset $S$ of nodes which intersect all graphs $F$ from a fixed graph family $\mathcal{F}$ that appears as subgraph in $G$.
A prominent example in this direction is the fundamental {\sc Feedback Vertex Set (FVS)} problem, where $\mathcal F$ is the class of all cycles.  FVS is one of Karp's 21 $\mathsf{NP}$-complete problems~\cite{Karp1972}.
It admits a 2-approximation in polynomial time~\cite{BafnaBF1999,BeckerG1996}, which cannot be improved to a $(2-\varepsilon)$-approximation for any $\varepsilon > 0$ assuming the Unique Games Conjecture~\cite{KhotR2008}.

Recently, several graph transversal problems have been revisited in the presence of additional parity constraints \cite{LokshtanovR2012,MisraRRS2012,LokshtanovRSZ2020,NageleZ2020}.
The natural parity variants of FVS are {\sc Odd Cycle Transversal} (OCT) and {\sc Even Cycle Transversal} (ECT), where one wishes to intersect the odd-length and even-length cycles of the input graph~$G$, respectively.
The approximability of these problems is much less understood than that of FVS: for OCT, only an $\mathcal{O}(\sqrt{\log n})$-approximation is
known~\cite{AgarwalCMM2005}, and for ECT, only a 10-approximation is known~\cite{MisraRRS2012}.

Planar graphs are a natural subclass of graphs in which to consider graph transversal problems.
The interest goes back to Baker's shifting technique~\cite{Baker1994}, which yielded a PTAS for {\sc Vertex Cover} in planar graphs (where $\mathcal F$ is the single graph consisting of an edge).
The technique was generalized by Demaine et al.~\cite{DemainBidimEPTAS}, who gave EPTASs for graph transversal problems satisfying a certain bidimensionality criterion, including FVS in \emph{unweighted} planar graphs.
That result was later extended to yield an EPTAS for FVS in unweighted $H$-minor free graphs~\cite{FominLST2010}, for any fixed graph~$H$.
In \emph{edge-weighted} planar graphs, PTAS are known for edge-weighted {\sc Steiner Forest} and OCT~\cite{BateniPlanarSteiner,PlanarMaxcut1,PlanarMaxcut2}.

On \emph{node-weighted} planar graphs, the situation appears to be
more complex.  First, the existence of a PTAS for FVS on node-weighted
planar graphs was a long-standing open question which was resolved
only recently in a paper of Cohen-Addad et
al.~\cite{CohenAddadEtAl2016}. The authors presented a PTAS for FVS in
node-weighted planar graphs, crucially exploiting the fact that the
treewidth of $G - S$ is bounded for feasible solutions $S$.  The
existence of an EPTAS for FVS in node-weighted planar graphs is still
open.

To deal with cycle transversal problems (in node-weighted planar graphs)
which are more complex than FVS, Goemans and
Williamson~\cite{GoemansW1998} first proposed a primal-dual based
framework.  Their framework requires the cycle family $\mathcal F$ to
satisfy a certain uncrossing property. The latter property can be
seen to be satisfied by OCT, {\sc Directed FVS} in directed planar
graphs, and {\sc Subset FVS}, which seeks a minimum-cost node set
hitting all cycles containing a node from a given node set~$T$. For
those problems, the authors obtained
3-approximations\footnote{18/7-approximations were claimed but later
  found to be incorrect~\cite{BermanY2012}.}.
  The framework of Goemans and Williamson~\cite{GoemansW1998} also yields a 3-approximation for {\sc Steiner
  Forest} in node-weighted planar graphs
\cite{SteinerDemaine,Moldenhauer2011}.  Berman and Yaroslavtsev
\cite{BermanY2012} later improved the approximation factor for the same class of uncrossable cycle transversal problems
from 3 to 2.4.  For none of those problems, though, the existence of a
PTAS is known.

The main question driving our work is whether the framework of Goemans and Williamson~\cite{GoemansW1998} (and its improvement by Berman and Yaroslavtsev~\cite{BermanY2012}) can be extended to cycle transversal
problems that do not satisfy uncrossability. In this paper we focus
on ECT in node-weighted planar graphs as a natural such problem: even
cycles are not uncrossable, and hence the framework of Goemans and Williamson~\cite{GoemansW1998} does not apply.  Furthermore, the
framework of Cohen-Addad et al.~\cite{CohenAddadEtAl2016} requires
that contracting edges only reduces the solution value, which is not
the case for even cycles either.  For \emph{unweighted} planar graphs, it is
still possible to obtain an EPTAS for ECT, by building on the work of
Fomin et al.~\cite{FominLRS2011}.  Their main result are EPTASs for
bi\-dimensional problems, which ECT is not (as contracting edges can
change the parity of cycles).  Yet, to obtain their result, they show
that any transversal problem that satisfies the
``$\nu$-transversability'' and ``reducibility'' conditions has an
EPTAS on $H$-minor free graphs (cf. \cite[Theorem~1]{FominLRS2011}).
Both conditions are met by unweighted
ECT\footnote{$\nu$-transversability follows from as graphs without
  even cycles have treewidth $2$, and reducibility from unit weights
  and con\-nectedness of the to-be-hit subgraphs $F$.}, which thus
admits an EPTAS on $H$-minor free graphs.  For ECT on node-weighted
planar graphs, though, reducibility fails, and the existence of a PTAS
is unknown.
The currently best result for ECT is a 10-approximation, which was given by Fiorini et al.~\cite{FioriniJP2010} for general graphs.
They showed that the integrality gap of the standard covering LP relaxation for ECT is $\Theta(\log n)$, but that adding sparsity inequalities reduces the integrality gap to~10.
No better than 10-approximation is known for ECT in node-weighted planar graphs.

\subsection{Our results}
We prove an improved approximation algorithm for ECT in node-weighted planar graphs.
\begin{theorem}
\label{thm:apxmain}
  ECT admits an efficient $47/7\approx6.71$-approximation on node-weighted planar graphs.
\end{theorem}
This improves the previously best 10-approximation by Fiorini et al.~\cite{FioriniJP2010} for planar graphs.

Our algorithm takes as input a node-weighted planar graph $G$ with
node weights $c_v\in\mathbb N$ for each $v\in V(G)$.  We then employ
a primal-dual algorithm that is based on the following natural covering LP for
ECT and its dual, where \C\ denotes the set of even cycles in $G$:

\begin{center}

\begin{minipage}{0.45\textwidth}   
  \leqnomode
  \vspace{-1em}
  \begin{align}
                 ~~~~~\min~& c^Tx\notag\\
    ~~~~~~\textnormal{s.t.}~&	x(C) \geq 1 \quad \forall \ C \in \mathcal C \label{lp:p2}\tag{$\mbox{P}_{  \textnormal{ECT} }$}\\
    	                & x \geq \0\notag
  \end{align}
\end{minipage}
\hfill\vline\hfill
\begin{minipage}{0.45\textwidth} 
  \vspace{-1.5em}
  \begin{align}
                 \max~& \mathbbm{1}^Ty\notag\\
    \textnormal{s.t.}~&	\sum\limits_{C \in \mathcal C, v \in C} y_C \leq c_v \quad \forall v \in V(G) \label{lp:d2}\tag{$\mbox{D}_{  \textnormal{ECT} }$}\\
    	                & y \geq \0\notag
  \end{align}
\end{minipage}

\end{center}

Fiorini et al.~\cite{FioriniJP2010} proved that the integrality gap of this LP is $\Theta(\log n)$.
Our main result is an improved integrality gap of this LP for ECT in planar graphs:
\begin{theorem}
\label{thm:weighted_apx}
 The integrality gap of the LP \eqref{lp:p2} is at most $47/7\approx6.71$ in planar graphs.
\end{theorem}

\subsection{Our approach}
Designing a primal-dual algorithm is far from trivial, as the imposed parity constraints rule out a direct application of the framework proposed by Goemans and Williamson~\cite{GoemansW1998}.
Unlike in their work, \emph{face-minimal even cycles} (even cycles containing a minimal set of faces in their interior) are not necessarily
faces, and may thus overlap.
Indeed, increasing the dual variables of face-minimal even cycles does not yield a constant-factor approximation in general. 

Consider \autoref{PieceCounterex}, and let $F$ be the inner face that is only incident to  blue and black nodes.
\begin{wrapfigure}{r}{.5\textwidth}
  \begin{center}
    \includegraphics[width=.5\textwidth]{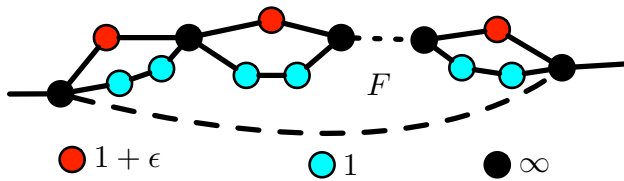}
    \caption{The bottom path has odd length, and the number of
      length-5 faces at the top is even.\label{PieceCounterex}}
  \end{center}
\end{wrapfigure}
For an even number of 5-cycles surrounding~$F$, $F$ is the only face-minimal even cycle in the graph.
Using only $F$ for the dual increase, even including a reverse-delete step, leaves one blue node of each 5-cycle. 
Yet, an optimal solution would take a single red and blue node from one 5-cycle.

To circumvent this impediment, we establish strong structural properties of planar graphs related to ECT. 
Those properties along with results from matching theory allow us to algorithmically find a large set of pairwise face-disjoint even cycles whose dual variables  we can then increment.
Even with this set of cycles found, it remains technically challenging to bound the integrality gap. 
For this purpose, we first use the structure of minimal hitting sets of  our graph to associate each such set with a hitting set in a subdivision of the so called 2-compression of our graph; the latter is a certain minor that we define in detail shortly. 
We then show  that faces that are contained in even cycles we increment are incident to few nodes on average.
Crucial in this step is a   technical result that is implicit in the work of Berman and Yaroslavtsev~\cite{BermanY2012}. 
Eventually, this approach leads to an integrality gap of $47/7$, and an algorithm with the same approximation guarantee.

\section{Primal-dual algorithm for ECT on node-weighted planar graphs}
\label{sec:pd}
We describe a primal-dual, constant-factor approximation for ECT on node-weighted planar graphs.
Our algorithm borrows some ideas from Fiorini et al.~\cite{FioriniJP2010} for the {\sc Diamond Hitting Set} (DHS) problem, which seeks a minimum-cost set of nodes in a node-weighted graph $G$ that hits all {\em diamonds} (sub-divisions of the graph consisting of three parallel edges). 
For DHS, Fiorini et al.~\cite{FioriniJP2010} employ a primal-dual algorithm to prove that the natural covering LP \eqref{lp:p2} (where $\mathcal C$ is replaced by the set of diamonds) has integrality gap $\Theta(\log n)$. 
We develop several new ideas to obtain a constant integrality gap.

We now outline the ideas of our primal-dual approach.
Consider a planar input graph $G$ with node costs $c_v\in\mathbb N$ for each $v\in V(G)$.
Given feasible dual solution~$y$ to \eqref{lp:d2}, let the \emph{residual cost} of node $v \in V(G)$ be $ c_v - \sum_{C \in \mathcal C, v \in C} y_C$.
Our primal-dual method begins with a trivial feasible dual solution $y=\0$, and the empty, infeasible hitting set $S = \emptyset$.

Then, in each iteration, we increase $y_C$ for all $C$ in some carefully chosen subset $\mathcal{C}'\subseteq \mathcal C$ of even cycles, while maintaining dual feasibility, and until some \emph{primary condition} is achieved. 
A common such primary condition is that some dual node-constraint becomes \emph{tight} in the increase process, and hence the corresponding node ends up having residual cost~$0$.

When this happens, we add the node to $S$.
Once $S$ is a feasible ECT, our algorithm ends its first phase, and executes a problem-specific \emph{reverse-delete} procedure.  
Here, we consider all nodes in~$S$ in reverse order of addition to $S$, and we delete such a node if the feasibility of~$S$ is maintained.
We will later describe a subtle and crucial refinement of this reverse-delete procedure.
Call the resulting final output of the algorithm $S'$.

During our algorithm, we will use the term \emph{hitting set} to refer to $S$, and during the analysis we will use the term \emph{hitting set} to refer to $S'$.
We will say a hitting set is \emph{feasible} if it is a feasible ECT, and refer to nodes of the hitting set as \emph{hit nodes}.

In the next subsections, we will fill in the details of the algorithm, and analyze the cost of~$S'$ compared to the value of an optimal
solution.
We begin by defining the concept of ``blended inequalities'' and necessary graph compression operations. 
Blended inequalities were used by Fiorini et al.~\cite{FioriniJP2010}, and our definitions follow their's closely. 

\subsection{Blended inequalities and compression}
A \emph{block} of $G$ is an inclusion-maximal 2-connected subgraph of $G$.
The \emph{block graph} of $G$ is the bipartite graph $B_G$ with bipartition $V(B_G) = B_1 \cup B_2 $, where $B_1$ are the blocks of~$G$, $B_2$ are the cut nodes of $G$, and $(b_1,b_2) \in B_1 \times B_2$ is an edge if $b_2$ is a node of $b_1$.

Let $S$ be a partial solution to the given ECT instance at some point during the execution of our algorithm. 
Then let $G^S$ be the corresponding residual graph that we obtain from $G - S$ by deleting all nodes that do not lie on even cycles. 
Our primal-dual algorithm now first looks for an even cycle $C$ in $G^S$ such that at most two nodes of $C$ have neighbours outside $C$. 
If such a cycle $C$ is found, we increment its dual variable $y_C$ until a node becomes tight.  
The reason for doing this is that such a cycle will pay for at most two hit nodes, which we will show later.

If there is no even cycle $C$ in $G^S$ such that at most two nodes of
$C$ have neighbours outside~$C$, we successively compress the residual
graph $G^S$ using two types of graph compression.  To this end, first
note that any minimal solution will only contain one node in the
interior of any induced path in $G^S$.  Suppose we contract some path
$P$ of $G^S$ of length at least two down to an edge $e$.  Choosing a
node in the interior of $P$ is ``equivalent'' to choosing the edge
$e$.  This is the motivation for the \emph{$1$-compression}.

Suppose we contract two $u$-$v$ paths $P_1,P_2$ with lengths of different parity down to edges~$e_1,e_2$, respectively.
We will find it helpful to think of these edges as a single
\emph{twin} edge between~$u$ and $v$ whose parity is {\em flexible}. 
This is the motivation for the \emph{$2$-compression}. 

Formally, we will successively compress $G^S$ as follows:
\begin{itemize}
  \item Obtain the {\em $1$-compression} $G^S_1$ of $G^S$ by repeatedly \emph{folding} degree-2 nodes $v$, as long as they exist, which means to delete $v$ and adding the edge $uw$ between its neighbors $u,w$.
 \item Note that no pair of nodes in $G^S_1$ is connected by more than two edges. 
   Obtain~$\bar{G}^S_1$ from~$G^S_1$ by replacing each pair of parallel edges by a {\em twin} edge. 
   In $\bar{G}^S_1$, we now once again fold degree-2 nodes as long as those exist.
   The resulting graph is the {\em $2$-compression} $G^S_2$ of~$G^S$.
\end{itemize}

See \autoref{2compress} for examples of 1- and 2-compression of a graph. 
In the following, we will omit the superscript $S$ from $G_1^S$, $\bar{G}^S_1$, and $G^S_2$ if this is clear from the context.
Let $G_3$ be obtained from~$G_2$ by replacing every edge of $G_2$ with a path of length two. 
If a twin edge was replaced, call the two edges of the path added \textit{twin edges}.  
By an abuse of notation, we say that a cycle of $G_1, \ G_2$ or~$G_3$ is even if it contains a twin edge, or if its preimage in $G$ is even.
\begin{figure}[h]
  \begin{center}
  \hspace{-3mm}  \includegraphics[scale=0.7]{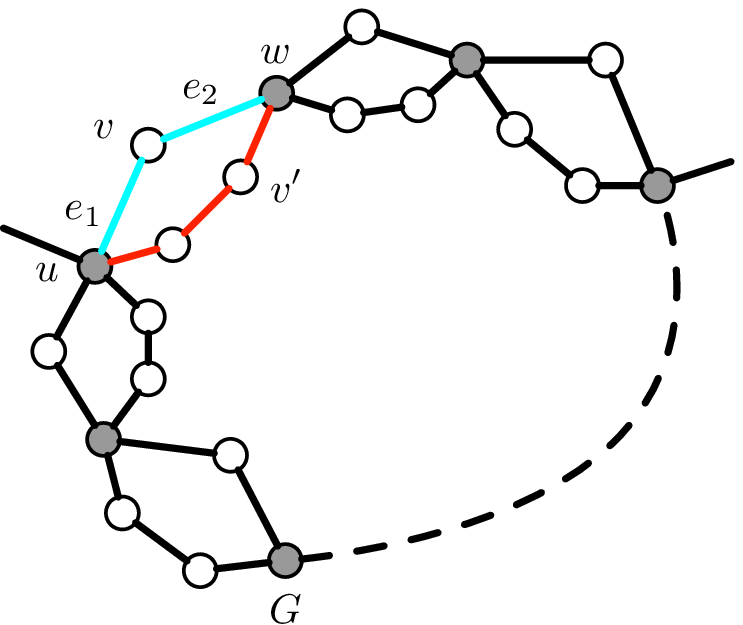}\hspace{-2mm}
    \includegraphics[scale=0.7]{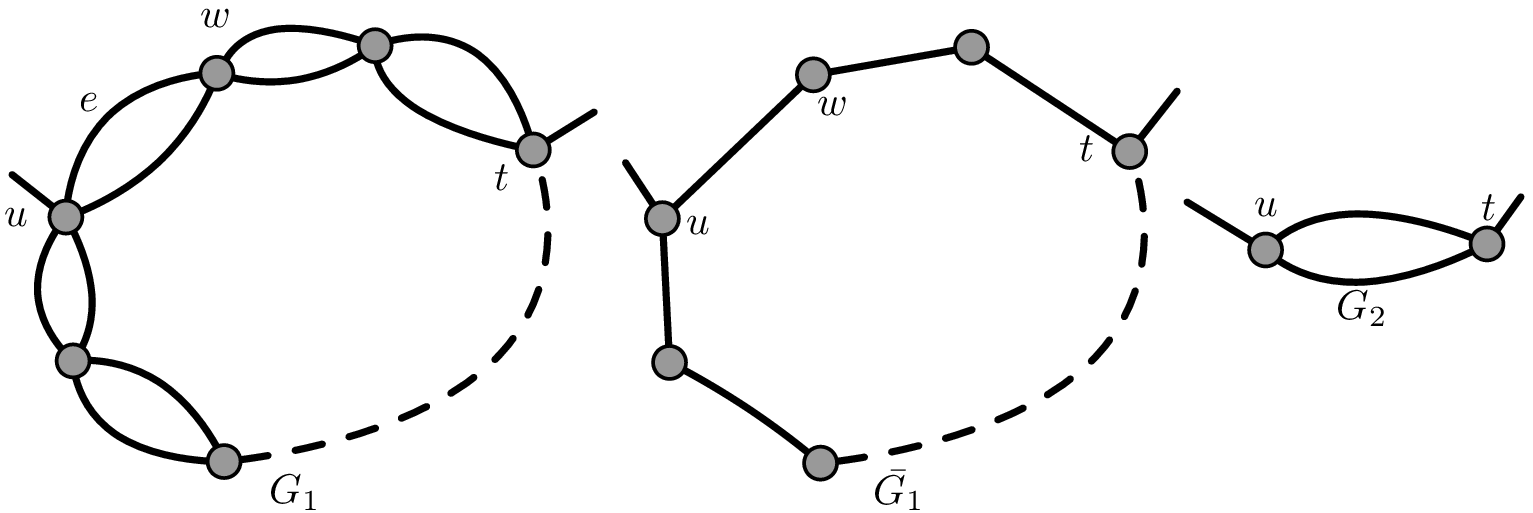}\hspace{-3mm}
    \caption{The graph $G$ and its 1- and 2-compression $G_1$ and $G_2$.\label{2compress}}
  \end{center}
\end{figure}

In the following, we will sometimes call the subgraph $Q$ of $G$ whose contraction yields a subgraph~$R$ of $G_2$ the {\em preimage} of $R$. 
If $R$ is an edge, call $Q$ a {\em piece}, and say $Q$ \emph{corresponds} to $R$. 
Furthermore, call $u$, $v$ \emph{ends} of $Q$ and other nodes of $Q$ \emph{internal nodes}.
If the edge was twin, call the piece \emph{twin}, otherwise, call the piece \emph{single}.
The blocks of a piece are cycles and paths, and the block graph of a piece is a path. 
Each cycle of a piece is called an {\em elementary cycle}.
For an elementary cycle $C$, call its two nodes $u_C$ and $v_C$ with neighbours outside $C$ {\em branch nodes}.
Call the two $u_C-v_C$-paths $P_1,P_2$ in $C$ the {\em handles} of $C$, which form the \emph{handle pair} $(P_1,P_2)$.
For an illustration, see the red and light blue edges in \autoref{2compress}.

The reason for defining $G_3$ is that intuitively selecting a node inside a piece corresponds to selecting the edge corresponding to the piece in $G_2$.
It will be simpler for us if our hitting set consists of only nodes, so we subdivide each edge of $G_2$.
Suppose that $S$ is the partial (and infeasible) hitting set for the cycles in \C\ at some point during the algorithm. 
Further, assume that~$G^S$ has even cycles, but none with at most two outside neighbours.
In this case, one can see that if an even cycle $C'$ in $G^S$ contains an internal node of some piece~$Q$, then $C' \cap Q$ is a path between the ends of $Q$.  
We illustrate this in \autoref{EvenCycInG}.
It follows that $C'$ has the form $ v_1 P_1 v_2 P_2 \hdots v_k P_k v_1 $, where for $i =1,\hdots,k$ the nodes $v_i,v_{i+1\mod k}$ are ends of some piece $Q_i$, and~$P_i$ is a $v_i$-$v_{i+1}$ path in $Q_i$.
For $i = 1,\dots,k$, the pieces $Q_i, Q_j$ 
for $i \neq j$ are disjoint except for their ends.
We will say that $C'$ in $G^S$ \emph{corresponds} to the cycle $C = (v_1,\hdots,v_k)$ in $G^S_2$.

\begin{wrapfigure}{r}{.35\textwidth}
  \begin{center}
    \includegraphics[scale=0.65]{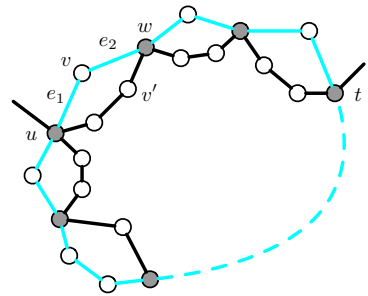}
    \caption{The light blue cycle in $G$ has two $u$-$t$ paths lying
      in different pieces of $G$; the dashed path has odd length.\label{EvenCycInG}}
  \end{center}
\end{wrapfigure}

For each such cycle $C$, its \emph{blended inequality} is
\begin{equation}
\label{eq:beq}\tag{$\circledast$}
  \sum_{v} a^C_v x_v \geq 1,
\end{equation}
where $a^C_v \in \{0, 1/2, 1\}$ for all nodes $v$, and where the support of $a^C$ is contained in the node set of the preimage of $C$. 
We next provide a precise definition of the coefficients of \eqref{eq:beq}.
With those, one can show that \eqref{eq:beq} is dominated by a convex combination of inequalities $x(C)\geq 1$ in \eqref{lp:p2}.

Consider an elementary cycle of the preimage of $C$ and let $h_1,h_2$ be its two handles. 
For each of these handles, we define its residual cost as the smallest residual cost of any of its internal nodes. 
Suppose that the residual cost of $h_2$ is at most that of $h_1$.
We will also call $h_1$ the {\em dominant}, and~$h_2$ the {\em non-dominant} handle of this cycle.   
As an invariant, our algorithm maintains that the designation of dominant and non-dominant handles of an elementary cycle does not change throughout the algorithm’s execution.

Suppose first that the residual cost of $h_1$ is strictly larger than that of $h_2$. 
In this case, let $a^C_v = 1$ for all internal nodes of handle $h_1$, and let $a^C_v = 0$ of the internal nodes of $h_2$.
If the residual cost of both handles is the same, we let $a^C_v = 1/2$ on internal nodes of both handles.

In certain cases, we need to {\em correct the parity} of the constructed inequality.
This is necessary if~$a^C$ as defined above is $0,1$ (i.e., if all elementary cycles of $C$ have a strictly dominant handle), and if the cycle formed by all dominant handles is odd.
In this case, we pick an arbitrary elementary cycle on $C$, and declare it {\em special}. 
For this special cycle, we then set $a^C_v=1$ for the internal nodes on {\em both} handles.
Following the same reasoning as Fiorini et al.~\cite{FioriniJP2010} for DHS, we can show the following for ECT:
\begin{lemma}
  Each feasible point of our LP \eqref{lp:p2} satisfies any blended inequality.
\end{lemma}
In our algorithm, we assume that inequalities \eqref{eq:beq} are part of \eqref{lp:p2}. 
Throughout the algorithm, we increase dual variables $y_{\circledast}$ of such inequalities. 

We will sometimes say that variable $y_{\circledast}$ (or cycle~$C$) \emph{pays for} $\sum_{v \in S'} a^C_v$ hit nodes.
It is well-known (see, e.g., Goemans and Williamson~\cite{GoemansW1998}) that if during any iteration dual variables for a family of blended
inequalities are incremented uniformly, and the dual variables pay for $\alpha$ hit nodes (of $S'$) on average, then the final solution produced by the algorithm is $\alpha$-approximate. 

The motivation for blended inequalities is to pay for no more than one
node in each piece.  Consider the example in \autoref{PieceCounterex}.
Here, the bottom black dashed path is odd, there are an even number of
handle pairs in the top part, and $\varepsilon$ is small.  Suppose we
set $a^C_v= 1/2$ on internal nodes of each handle.  If we were to
increment the inequality \eqref{eq:beq}, all the blue nodes of
weight~$1$ would become tight, and after reverse-delete, the algorithm
would keep one blue node for each handle pair.  However,
selecting a red node and a blue node would be a
cheaper solution.  This could be achieved by setting $a^C_v = 1$ for
red and black nodes, and $a^C_v = 0$ on blue nodes,
until the residual costs of the red nodes become 1, and
afterwards setting $a^C_v = 1/2$ on internal nodes of each handle.

During its execution, the algorithm carefully chooses a family of
even cycles \C\ in $G^S_2$ and increments the dual variables of
certain blended inequalities for each $C \in \C$ until a node becomes
tight, or the blended inequality changes; i.e. the residual costs of
two handles of a handle pair, which were previously not equal, become
equal.

In their primal-dual algorithms for cycle transversal problems with
uncrossing property, Goemans and Williamson~\cite{GoemansW1998}
started with the infeasible ``hitting set'' $S= \emptyset$.  While~$S$
is infeasible, the dual variables for faces of the residual digraph
that are cycles are incremented.  A reverse-delete step is applied at
the end. The authors show that tight examples for their algorithm
feature so called {\em pocket} subgraphs.  Not surprisingly, the
improved algorithm of Berman and Yaroslavtsev~\cite{BermanY2012} has
to pay special attention to these pockets to obtain the improvement in
performance guarantee.

\subsection{Pockets and their variants}
The following definition of crossing cycles was elementary to the
approach by Goemans and Williamson~\cite{GoemansW1998} for cycle
transversal problems in planar graphs.

\begin{definition}
\label{cross}
  In an embedded planar graph,  two cycles $C_1,C_2$ \emph{cross} if $C_i$ contains an edge intersecting the interior of the region bounded by $C_{3-i}$, for $i=1,2$.
  That is, the plane curve corresponding to the embedding of the edge in the plane intersects the interior of the region of the plane bounded by~$C_{3-i}$. 
  A set of cycles $\mathcal{C}$ is \emph{laminar} if no two elements of $\mathcal{C}$ cross. 
\end{definition}

Next, we formally define pockets, and
we also introduce the new notion of ``pseudo-pockets'', the lack of which will help us ``cover'' our graph with even cycles.

\begin{definition}
  Let $G$ be a graph and let $\mathcal C$ be a collection of cycles in $G$.
  A \emph{pseudo-pocket} of $(G,\mathcal C)$ is a connected subgraph $G'$ of $G$ which contains a cycle such that at most two nodes of $G'$ have neighbours outside $G'$. 
  A \emph{pocket} of $(G,\mathcal C)$ is a pseudo-pocket that contains a cycle of $\mathcal{C}$.
  A pocket is \emph{minimal} if it contains no pocket as a proper induced subgraph.
\end{definition}
\begin{figure}[h]
  \centering
  \includegraphics[scale=0.8]{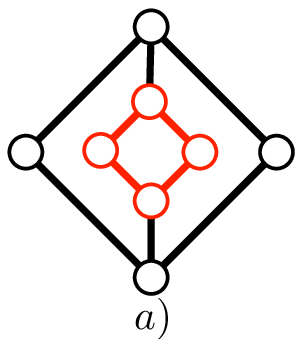}
  \includegraphics[scale=0.8]{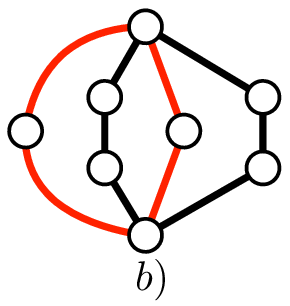}
    \caption{(a) Graph formed by red nodes is a
    pocket. (b) Crossing cycles in red and black.\label{pocket}}
\end{figure}

\subsection{Identifying families of even cycles via tilings}
The $12/5$-approximation algorithm of Berman and Yaroslavtsev~\cite{BermanY2012} for {\sc Directed FVS} in node-weighted planar digraphs $G$ proceeds roughly as follows. 
 
It starts with the empty hitting set $S=\emptyset$.
As long as $S$ is not a hitting set for the directed cycles of $G$, it first looks for a pocket $H$ of the residual digraph $G^S$, that is the digraph obtained from $G-S$ by deleting all nodes not on a directed cycle. 
It then increments the dual variables for the set of face minimal directed cycles of $H$, which happen to be faces. 
It then adds any nodes that become tight to $S$. 
Once $S$ is feasible, the algorithm performs a reverse deletion step.
 
As pointed out, in our setting, face-minimal even cycles may not be faces, and may cross.
Following Berman and Yaroslavtsev~\cite{BermanY2012},  we wish to ``cover'' our residual graph with face-minimal even cycles which do not cross, we call this a ``tiling''; see \autoref{tiletomat} iii).
As we will see, this tiling allows us to identify the dual variables to increase. 
Let us formalize the correspondence between edges of the dual between odd faces and even faces.

\begin{definition}\label{match-tile}
  Let $H$ be a plane graph without pseudo-pockets. 
  For each face $f$ of $H$, let~$v_f$ be the corresponding node of the planar dual $H^*$.
  A \emph{tile} of $H$ is an even cycle $C$ of $H$ bounding one or two faces.
  If $C$ is a single face $f$, we say that $C$ \emph{corresponds} to the node $v_f$.
  If $C$ bounds two faces $f$ and $g$, we say that $C$ corresponds to the edge $v_f v_g \in E(H^*)$.
  We say that nodes $v_f,v_g$ and the faces $f,g$ are \emph{covered} by the tile.
\end{definition}
For a node $v$ of $H^*$, let $f_v \subset E(H) $ be the edges on the
boundary of the corresponding face of~$H$.
Denote by~$h_\infty$ the node of $H^*$ corresponding to the infinite face.

Given $ wh_\infty \in E(H^*)$, a cycle $C^1 \subset E(H)$ \emph{corresponds} to $w v_\infty$ if $C^1$ is a cycle of $f_w \Delta f_{h_\infty}$, or $C^1= C' \Delta f_w$ and $C'$ is a cycle of $f_w \Delta f_{h_\infty}$.
We also call such a cycle $C^1$ a \emph{tile} and say that $C^1$ covers $h_\infty$, $w$, and the corresponding faces.

Given a matching $E' \subset E(H^*)$ and $V' \subset V(H^*)$, with $E'= \{e_1,\hdots,e_\ell\}$ and $V'=\{v_1,\hdots,v_t\}$, a set of tiles $ \calT = \{C_1,\hdots,C_{\ell+t}\}$ \emph{corresponds} to $E' \cup V'$ if $C_i$ corresponds to $e_i$ for $i = 1,\hdots,\ell$ and~$C_{j+\ell}$ corresponds to $v_j$ for $j = 1,\hdots,t$.

In \autoref{tiletomat} i), cycle $C$ bounds two faces $f$ and $g$; see also \autoref{tiletomat} ii).

\begin{figure}[h]
  \centering
  \includegraphics[scale=0.8]{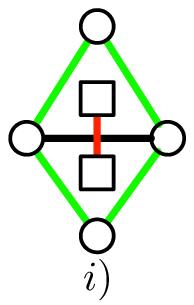}
  \includegraphics[scale=0.8]{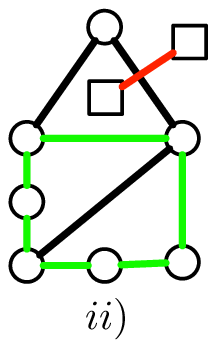}
  \includegraphics[scale=0.8]{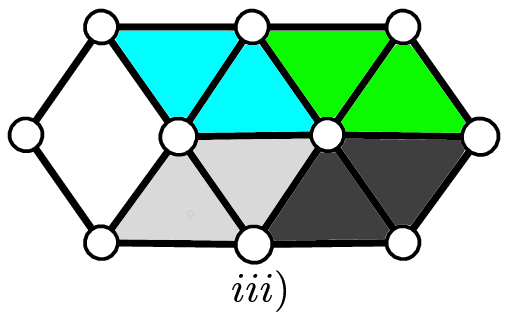}
  \includegraphics[scale=0.8]{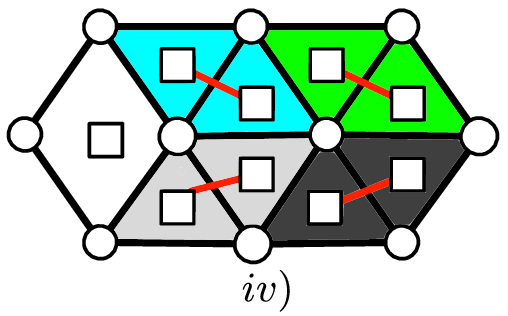}
  \caption{Diagrams i) and ii) show cycles in green and corresponding edges of the dual graph in red.
      (i) The red edge corresponds to the symmetric difference of two finite faces. 
      (ii) The red edge corresponds to the symmetric difference of a finite and infinite face.
      Diagrams iii) and iv) show a tiling indicated by the boundaries of the  various finite regions in white, light grey, etc and the corresponding matching.
    \label{tiletomat}}
  \end{figure}
\begin{definition}
  For a plane graph $H$, a set $\calT$ of tiles is a \emph{pseudo-tiling} if no face of $H$ is covered by more than one tile.
  If the node $ v_{h_\infty} $ corresponding to the infinite face of $H$  is not covered by $\calT$, we call $\calT$ a \emph{tiling}. 
\end{definition}

Certain tilings are particularly desirable; we will define these the next. 
\begin{definition}
\label{quasiperdef}
  Let $\alpha\in(0,1)$.
  A tiling is \emph{$\alpha$-quasi-perfect} if it covers all even finite faces, a $\beta$-fraction of odd finite faces of~$G^S$, and a
  $\psi$-fraction of the finite faces of $G^S$ are even, where 
  \begin{equation}
    \beta (1-\psi)+2\psi \geq \alpha.
  \end{equation}
\end{definition}

Let $C$ be an even cycle in $G^S_2$, and recall that we say that $C$
pays for $\sum_{v \in S} a^C_v$ hit nodes.  For an even cycle in a
tiling consisting of two faces, we bound the number of hit nodes it
pays for by the number of hit nodes each face pays for.

We will show that a finite face of our graph intersects at most $18/7$
hit nodes on average (over all finite faces). 
Ideally, we would want to cover all faces by a tiling.  
Then an even cycle of our tiling is incident to at most $36/7$ hit nodes on
average, twice the amount a face of our graph intersects on average.
Alas, tilings covering all faces need not always exist.
Thus, we try to find a tiling that covers as many finite faces as possible.  
Suppose that we find a tiling $\calT$ that covers a set $ \calT_{\textsf{Faces}} $ of finite faces consisting of $\alpha$-fraction of the finite faces of our graph.  
It follows that a face of $ \calT_{\textsf{Faces}} $ will be incident to
at most $18/7\alpha$ hit nodes on average, and so an even cycle of the
tiling $\calT$ is incident to at most $36/7 \alpha$ hit nodes on
average.  
Intuitively, even faces pay for fewer hit nodes than even cycles containing two faces, so it is good if a tiling contains many even faces. 
The motivation for quasi-perfect tilings is that it is good if a large fraction of faces are covered by the tiling and if the tiling contains a lot of even faces. 
We prove the following key result in \autoref{2o3quasitil}.

\begin{restatable}{theorem}{thmcmd}\label{2o3quasperexist}
Let $H$ be a $2$-compression of some planar graph $G$  that has an even cycle and contains no pockets.
Then $H$ has a $2/3$-quasi-perfect tiling.
\end{restatable}

\subsection{The algorithm in detail}

We can now formally state our algorithm.  It takes as input a planar
graph $G$ with cost function $c:V(G)\rightarrow\mathbb N$.  Let
$\C(G)$ denote the set of even cycles of~$G$, and let
$\mathsf{opt}(G,c)$ denote the minimum cost of an \emph{even cycle
  transversal} of~$G$, which is a set of nodes intersecting every
cycle in $\mathcal C(G)$.

As we will see, the algorithm returns an even cycle transversal $S$ of
$G$ whose cost is at most $(47/7)\mathsf{opt}(G,c)$.  We start with
the empty candidate $S:= \emptyset$.  In each iteration, the
algorithm looks for an even cycle $C$ in the residual graph
$G^S$ such that at most two nodes of~$C$ have outside neighbours.  If
we find such $C$, increment the variable~$y_C$ until a node becomes
tight.  If no such cycle exists, the algorithm computes the
2-compression of~$G^S$, and in it, we find an inclusion-minimal pocket
$H$ of $G^S_2$.  Using \autoref{2o3quasperexist}, we find a
$2/3$-quasi-perfect tiling~$\tilT_H$ of $H$ and increments the
dual variables for the blended inequalities for each
$C \in \tilT_H$.  The algorithm then adds all nodes $X$ that
became tight to our candidate hitting set $S$.

During an iteration, for each handle pair $(Q_1,Q_2)$ for which the set $X$ of nodes that became tight contains a node in the interior of each handle, our algorithm will choose two nodes $a,b \in X$ with $a$ in the interior of $Q_1$ and $b$ in the interior of $Q_2$ and define $(a,b)$ to be a \emph{node pair}.
For instance, in \autoref{2compress} if $v$ and $v'$ are the only
nodes added during some iteration then the algorithm would define
$(v,v')$ to be a node pair.  For a set of nodes $X$ added during the
same iteration, nodes in a pair are considered to be added {\em before} any
node not in a pair.

At the end of the algorithm, we perform a non-trivial reverse-delete procedure.
Formally, let $w_1,\hdots,w_\ell$ be the nodes of $S$ in the order they were added to $S$ by the algorithm, where for nodes $w_i,w_j$ that were added during the same iteration if $w_i$ is in a pair and $w_j$ is not, then $i<j$. 
That is, for reverse-delete purposes, nodes not in a pair are considered for deletion first.
For $p=\ell,\ell-1,\hdots,1$, if $w_p$ is not in a node pair, then if $S \backslash \{w_p\}$ is a feasible ECT, the algorithm deletes $w_p$ from $S$; otherwise, it does not. 
If $w_p$ is in a node pair $(w_p,w')$, then if $S \backslash \{w_p,w' \}$ is a feasible hitting set, then delete both $w_p,w'$ from $S$; else, keep both~$w_p,w'$.

\begin{algorithm}[ht]
  \caption{EvenCycleTransversal$(G,c)$\label{tiling}}
  \SetAlgoVlined
  \SetKwInOut{Input}{Input}\SetKwInOut{Output}{Output}
  \SetKw{KwDownTo}{downto}
  \Input{A planar graph $G$ with node costs $c:V(G)\rightarrow\mathbb N$.}
  \Output{An even cycle transversal $S$ of $G$ of cost at most $\frac{47}{7}\mathsf{opt}(G,c)$.}
  $S \leftarrow \emptyset$ \\
  \While{residual graph $G^S$ contains an even cycle}{
	  \eIf{$G^S$ contains a cycle $C$ with at most $2$ outside neighbours}{
		  increase the dual variable $y_C$ for $C$ until a node $v$ becomes tight.
  	}{
		compute the 2-compression $G^S_2$ of $G^S$.
	}
	$H\leftarrow$ minimal pocket of $G^S_2$.\\
  $\tilT_H \leftarrow$ a $2/3$-quasi-perfect tiling of $H$.\label{alg:qpt}\\
  Increment dual variables of blended inequalities of all $C \in \tilT_H $ until a node $v$ becomes tight or the blended inequality changes.\\
  Add all nodes that became tight to $S$. Denote by $X$ the set of nodes that became tight. 
  \For{ each handle pair $(Q_1,Q_2)$ }{ 
  \If{$X$ contains a node in the interior of each handle }{choose two nodes $a,b \in X$ with $a$ in the interior of $Q_1$ and $b$ in the interior of $Q_2$ and define $(a,b)$ to be a node pair.}{} }      
}
$w_1,\hdots,w_\ell\leftarrow$ nodes of $S$ in the order they were added, where for nodes $X$ added during the same iteration, any node of $X$ in a pair appears before others node of $X$ not in pairs.
\For{ $i = \ell$ \KwDownTo $1$}{
	\eIf{$w_i$ is not part of a pair}{ 
		\If{$S \backslash \{ w_i \}$ is feasible}{
			$S \leftarrow S \backslash \{ w_i \} $.
		}
	}{
		Let $(w_i,w_j)$ be the pair containing $w_i$.
		\If{$S \backslash \{ w_i, w_j \}$ is feasible}{
			$S \leftarrow S \backslash \{ w_i, w_j \} $.
		}
	}
}
\Return{$S$}
\end{algorithm}

The intuition behind the caveat in our reverse-delete step is that node pairs are often very useful to keep, because they disconnect a piece. 
Consider the example in \autoref{VPairCounterEx}. 
\begin{figure}[h]
  \centering
  \includegraphics[scale=0.8]{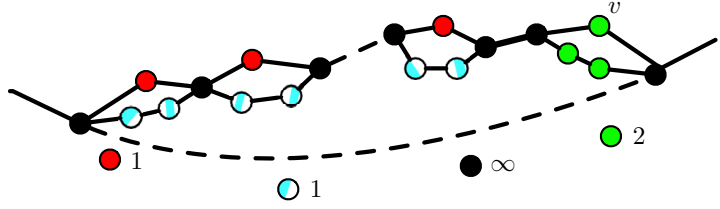}
  \caption{The red and blue striped nodes have weight $1$, black nodes have infinite weight and green nodes have cost $2$.
    The bottom dashed black path has odd length. 
    The number of length-5 faces at the top is assumed to be even.
    \label{VPairCounterEx}}
\end{figure}
There is a piece with green nodes of cost $2$, and an odd number of
length-5 faces with red and blue striped nodes of cost 1.  The black
nodes have cost infinity.  The bottom dashed path has odd length.  In
the 2-compression, all length-5 faces in the figure
belong to one piece.  Suppose for the blended inequality we choose the
length-5 face with the green nodes as the special cycle, and
we increment this blended inequality.  One 
sees that the red, blue striped and green nodes become tight
simultaneously.

To see that reverse delete orders need to be chosen carefully,
consider the following adversarial ordering: in reverse delete,
consider the two green nodes other than $v$ first, then consider the
red nodes, and then consider one blue striped node on each
handle. Finally, consider the remaining blue striped nodes.  One can
see that the algorithm would end up with $v$ and one blue striped node
per handle, which is significantly more costly than the optimum which
selects the solution consisting of one red and one blue striped node
on a handle pair.  This completes the description of our approximation
algorithm for ECT, whose complete pseudo-code is given as
\autoref{tiling}.

\subsection{Analysis of approximation ratio}
We claim that the algorithm is a $47/7$-approximation for ECT on node-weighted planar graphs.

Fix an input planar graph $G$ with node costs $c_v\in\mathbb N$.
Consider a set $S\subseteq V(G)$ of nodes and a node $v \in S$.  A
cycle $C$ is a \emph{pseudo-witness cycle} for $v$ with respect to $S$
if $C \cap S = \{v\}$.  If $C$ is additionally even, then $C$ is a
\emph{witness cycle} for $v$.  Note that if $S$ is an
inclusion-minimal ECT for~$G$, then there is a set $W_v$ of witness
cycles for each node in $v\in S$.  If the reverse-delete procedure
does not delete any node of $S$, then each node not in a pair has a
witness cycle and for each pair, at least one of the nodes in the pair has a
witness cycle.

The analyses of the algorithms by Goemans and Williamson~\cite{GoemansW1998}, and by Berman and Yaroslavtsev~\cite{BermanY2012}, for {\sc Subset FVS} on planar graphs rely
crucially on the fact that, each node of an inclusion-wise minimal
solution has a witness cycle.  Goemans and
Williamson~\cite{GoemansW1998} showed that one can find a laminar
collection
$\mathcal{A}$
of witness cycles. Laminar families are well-known to have a natural
tree representation. The key argument by Goemans and Williamson~\cite{GoemansW1998}, and by Berman and Yaroslavtsev~\cite{BermanY2012}, is that for each {\em leaf} cycle $C$
of the laminar family, 
one can increment the dual variable of at least one face contained in
the region defined by $C$. 
Further, this dual variable 
pays only for the hit node that $C$ is a witness of.  This is used to
argue that a large portion of the dual variables they incremented pay
for a single hit node.  An additional bound on how many nodes the
other dual variables pay for is proven exploiting the sparsity of
planar graphs.

For the ECT problem, however, we do not have laminar witness cycles.
Instead, we must extend the analysis of Berman and
Yaroslavtsev~\cite{BermanY2012} to find a set of laminar
pseudo-witness cycles.

Consider some time $\bar{t}$ during the algorithm when applied to $(G,c)$.
Let $S_{\bar{t}}$ be the current hitting set and $G^{S_{\bar{t}}}$ the residual graph.
Let $\{\sum_{v \in V(G)} a^C_v  \geq 1 \}_{C\in \mathcal{L}}$ be the set of inequalities of the increased dual variables.
Here, $\mathcal{L}$ will be either a single cycle of $G^{S_{\bar{t}}}$, or a tiling of $G^{S_{\bar{t}}}_2$.
We wish to show that the primal increase rate towards the final set $S'$ at time $\bar{t}$,   $ \sum_{C \in \mathcal{L} } \sum_{v \in S'} a^C_v $  is at most $47/7$ times the dual increase rate $|\mathcal{L}|$. 

If the algorithm incremented $y_C$, where $C$ was a cycle of $G$ for which at most two nodes have outside neighbours, then the inequality we increase is $ \sum_{v \in C} x_v \geq 1 $.
As $S'$ is minimal under reverse-delete, $|C \cap S'| \leq 2$, and hence the primal increase rate $\sum_{v \in S'} a^C_v = |C \cap S'| $ is at most twice the dual increase rate 1.

Otherwise, if the algorithm did not increment $y_C$, then there is no
cycle $C$ of $G^{S_{\bar{t}}}$ such that at most two nodes of $C$ have
neighbours outside $C$.  Hence, the set of increased inequalities are
the blended inequalities of a tiling $\tilT_H$ of an inclusion-minimal
pocket $H$ of ~$G^{S_{\bar{t}}}_2$.  For a cycle $C$
of~$G^{S_{\bar{t}}}_2$, let
$\sum_{v \in V(G^{S_{\bar{t}}})} a^C_v \geq 1$ be the blended
inequality $C$ (see \autoref{eq:beq}).

Recall that informally speaking, we wish to pay for at most one hit node inside a piece.
To do this, we need the following theorem which generalizes a result by Fiorini et al.~\cite[Theorem 5.7]{FioriniJP2010} and tells us the structure of a minimal solution within a piece.

\begin{theorem}
\label{node piece}
  Let $S'$ be the output of \autoref{tiling} on input $(G,c)$. 
  Consider an edge $uw \in E(G^{S_{\bar{t}}}_2)$ on the
  even cycle whose dual variable we increase, 
  and let $Q$ be the piece corresponding to $uw$ in~$G$.
  Then exactly one of the following occurs:
  \begin{enumerate}
    \item $S'$ contains no internal node of $Q$,
    \item $S'$ contains exactly one node of $Q$, and this node is a cut-node of $Q$,
    \item $S'$ contains exactly two nodes of $Q$, and they belong to opposite handles of a cycle of $Q$,
    \item $S'$ contains exactly one node per elementary cycle of $Q$, each belonging to the interior of some handle of the corresponding cycle.
  \end{enumerate}
\end{theorem}
\begin{proof}
  If $S'$ contains two nodes $a$ and $b$ in the interiors of different handles of a pair, then since removing both $a$ and $b$ disconnects $u$ from $w$ in $Q$, our algorithm would delete all other nodes of $V(Q) \backslash \{ u,w\}$ from $S'$.
  If $u$ or $w$ were in $S'$, then our algorithm would delete both $a$ and $b$.
  Thus $u,w \notin S'$, and case 3 holds.  

  Similarly, if $S'$ contains a cut node $z$, then since removing $z$ disconnects from $u$ from $v$ in $Q$, our algorithm would delete all other nodes of $V(Q) \backslash \{u,v\}$ from $S'$.
  If $u$ or $w$ were in $S'$, then our algorithm would delete $z$.
  Thus $u,w \notin S'$, and case 2 holds.

  If $u$ or $w$ is in $S'$, then for any $r \in S' \cap ( V(Q) \backslash \{ u,w \} ) $ there cannot be an even cycle of $G$ which intersects $S'$ only at $r$ as such a cycle would have to go through $u$ or $w$, and thus~$S'$ contains no internal node of $Q$ and case 1 holds. 

  Assume that cases 1,2 and 3 do not hold, so $u,w \notin S'$.
  Let $(P_1,P_2)$ be a handle pair on~$Q$ such that $P_1$ contains a hit node $t$ in its interior and $P_2$ does not. 
  Suppose that $Y_1,Y_2$ was another handle pair with no hit node on either of $Y_1$ or $Y_2$. 
  By our deletion procedure, there must be an even cycle $C$ which intersects $S'$ at $t$ only. 
  Such a cycle $C$ uses the handle $P_1$ and one handle $Y_i$ of the pair $Y_1,Y_2$. 
  Let $C'$ be the cycle obtained from $C$ by replacing the paths $P_1$ and $Y_i$ in $C$ by the paths $P_2$ and $Y_{3-i}$. 
  Since the lengths of different handles of a pair have different parity, $C'$ is even.  
  Since $P_2,Y_1$ and $Y_2$ contain no nodes of $S'$, $C'$ contains no nodes of $S'$, which is a contradiction. 
  Since a handle can only contain one hit node of $S'$, this implies that case 4 holds.
\end{proof}

Given a hitting set $S'$ 
output by  \autoref{tiling}, we wish to
construct a corresponding hitting set for $G^{S_{\bar{t}}}_3$ such
that the primal increase rate of any particular blended inequality (with
respect to~$S'$) is 
equals the number of nodes of $  S'_3$ on the
corresponding cycle of $G^{S_{\bar{t}}}_3$.
\begin{definition}
\label{G vs G3}
Let $S'$ be a hitting set output by \autoref{tiling}.  The
\emph{corresponding hitting set for~$G^{S_{\bar{t}}}_3$} is the set
$ S'_3 \subset V(G^{S_{\bar{t}}}_3)$ obtained by first taking the
nodes of $S' \cap V(G^{S_{\bar{t}}}_3)$. Now, consider an edge $uv$ of
$G^{S_{\bar{t}}}_2$ with corresponding piece $P$. 
Replace $uv$ by the path $u w_p v$ in $G^{S_{\bar{t}}}_3$, and add $w_p$ to
$ S'_3$ if $P - S'$ has two components.\footnote{Note that the
minimality of $S'$ implies that removing $S'$ from $P$ yields at most
two connected components.}
\end{definition}

\begin{claim}\label{blendG3}
  Let $C$ be the preimage of an even cycle in $G^{S_{\bar{t}}}_2$, and $C_3$ the corresponding cycle in~$G^{S_{\bar{t}}}_3$.
We claim that $ \sum_{v \in S'} a^C_v \leq |C_3 \cap  S'_3 | +1$. 
Further, if $C$ does not contain a twin edge, then it holds $ \sum_{v \in S'} a^C_v \leq |C_3 \cap  S'_3 |$.
\end{claim}
 \begin{proof}[Proof of \autoref{blendG3}]

Define $b^C$ as follows: For a handle pair, while one handle has greater residual cost than the other set $b^C_v=1$ for $v$ on the handle of greater residual cost $b^C_v=0 $ on internal nodes of the other handle (change $b^C$ whenever residual costs become equal). 
Otherwise, $b^C_v=1/2$ on internal nodes of both handles. In short,
$b^C_v$ are the coefficients $a^C_v$ if we had not redefined $a^C_v=1$
for nodes on the special cycle.

Let $uw \in E(G^{S_{\bar{t}}}_2) $, $Q$ be the preimage of $uw$ in $G^{S_{\bar{t}}}$ and $u w_Q w$ be the subdivision of $uw$ in $G^{S_{\bar{t}}}_3$.
Let $S'_3$ be the corresponding hitting set of $S'$ for $G^{S_{\bar{t}}}_3$. 
We claim $\sum_{v \in S' \cap (Q \backslash \{ u,w \} } b^C_v = |S'_3 \cap \{ w_Q \}| $.
We distinguish which case of \autoref{node piece} is satisfied by $uw$ and $S'$.
\begin{itemize}
  \item If $uw$ and $S'$ satisfy (1), then $\sum_{v \in S' \cap (Q \backslash \{ u,w \} } b^C_v = 0 $. Since $S'$ contains no internal node of $Q$, $Q \backslash S$ is connected and hence $S'_3$ does not contain $w_Q$. Hence $\sum_{v \in S' \cap (Q \backslash \{ u,w \} ) } b^C_v =  |S'_3 \cap \{ w_Q \}|$.
  \item If $uw$ and $S'$ satisfy (2) or (3), then $S'$ does not contain either end node of $Q$, and contains either a single cut node of $Q$, or exactly two nodes of $Q$ in the interiors of two handles of a handle pair of $Q$.
    Thus, $S' \cap Q$ consists either of a single node $v$ for which $b^C_v=1$, or two nodes $j,k$ for which $b^C_j=b^C_k=1/2$, and so $\sum_{v \in S' \cap Q} b^C_v = 1$. 
    
    In either case (2) or (3), $Q \backslash S'$ is disconnected so $|S'_3 \cap \{ w_Q \}| =1$. Hence $\sum_{v \in S' \cap (Q \backslash \{ u,w \} ) } b^C_v =  |S'_3 \cap \{ w_Q \}|$.
  \item Suppose $S'$ satisfies (4).  
  Suppose for a contradiction that \autoref{tiling} added a node pair $(l,m)$ 
  on some handle pair $(P_1,P_2)$  of $Q$. 
  It then follows from the reverse-delete step that the final solution $S'$ contains both $l$
and $m$, or none of them. 
Since we do not contain a node pair, the deletion procedure of \autoref{tiling} implies the algorithm did not
add a node pair with nodes in $Q$.  

Hence, throughout the algorithm, for each handle pair $(P_1,P_2)$  of $Q$, the handle $P_i$, which contains a hit node in its interior must have strictly less residual cost than the other.
Hence $b^C_v=0$ on handle $P_i$. This implies
\begin{equation}
  \sum_{v \in ( V(Q) \backslash \{ u,w \} ) }  b^C_v = 0 \enspace . \label{item 4} 
\end{equation}

 \end{itemize}
Thus, $\sum_{v \in S' \cap (Q \backslash \{ u,w \} } b^C_v = |S'_3 \cap \{ w_Q \}| $.

Let $C=v_1 v_2 \hdots v_\ell v_1$. Let $Q_i$ be the piece corresponding to $v_i v_{i+1 \mod \ell} $. 
Let $ q_i $ be the node resulting from subdividing $v_i v_{i+1 \mod l} $ in $G^{S_{\bar{t}}}_2$ to obtain $G^{S_{\bar{t}}}_3$.
Let $C_3:=v_1 q_1 v_2 , q_2 ,\hdots,  v_\ell q_\ell$ the cycle corresponding to $C$ in $G^{S_{\bar{t}}}_3$.  
We showed 
\begin{equation}\label{QiandS'}
    \sum_{v \in S' \cap (Q_i \backslash \{ u,w \} ) } b^C_v = |S'_3 \cap \{ q_i \}| \enspace .
\end{equation}
Summing \eqref{QiandS'} for $i-1,\hdots,\ell$ yields 
$ \sum_{v \in S' \cap (\cup_{i=1}^l Q_i \backslash \{ v_1,v_2,\hdots,v_\ell \}  ) } b^C_v = | \{ q_1,q_2,\hdots,q_\ell\} \cap C_3 | $.

Noting $b^C_{v_i}=1$ for each $i$ and $b^C_v=0$ for $v \notin \cup_{j=1}^\ell Q_j$, yields
  \begin{equation}
  \label{blendtoG3}
    \sum_{v \in S'} b^C_v = |C_3 \cap  S'_3 | \enspace . 
  \end{equation}
Let us now relate $a^C_v$ to $b^C_v$.
If $C$ has no twin edge, then the blended inequality coefficients $a^C_v$ are equal to $b^C_v$, therefore $\sum_{v \in S} a^C_v = |C_3 \cap  S'_3 | $. 

In general, $C$ may contain a twin edge. In this case, $a^C_v$ differs from $b^C_v$ only in the interior of the handles $H_1,H_2$  
of the special cycle: then either $b^C_v=\frac{1}{2}$ in the interior of $H_1$ and $H_2$, or $b^C_v=0$ in the interior of the dominant handle, and $b^C_v=a^C_v$ everywhere else.

If $b^C_v=\frac{1}{2}$ in the interior of $H_1$ and $H_2$, then note from \autoref{node piece} there are at most two nodes of $S'$ on $H_1 \cup H_2$.
Thus, $\sum_{v \in S} a^C_v \leq \sum_{v \in S} b^C_v + 1$.  

Otherwise, $b^C_v=0$ in the interior of the dominant handle, and $b^C_v=a^C_v$ everywhere else.
Since $S$ contains at most one node from the dominant handle $\sum_{v \in S} a^C_v \leq \sum_{v \in S} b^C_v + 1  $.
Thus,  $\sum_{v \in S} a^C_v \leq |C_3 \cap  S'_3 | +1$ completing the proof.
\end{proof}

To show that $|C_3 \cap  S'_3 | +1$ is small on average we need the
fact that $S'_3$ is a minimal ECT, which is stated in the following remark.

\begin{remark}
  Let $S'$ be the output of \autoref{tiling} on input $(G,c)$.
  Let $ S'_3$ be the corresponding hitting set for $G^{S_{\bar{t}}}_3$ in \autoref{G vs G3}.
  Then there is a witness cycle for each $v \in  S'_3$.
\end{remark}

For a node $h$ and cycle $C$, denote by $C \circ h$ that $h$ lies on $C$.
\begin{definition}
\label{debit}
  Let $\mathcal{R}$ be a set of cycles of a graph $G$, and let $S \subset V(G)$.
  The \emph{debit graph} for $\mathcal{R}$ and $S$ is the bipartite graph $\mathcal{D}_G = (\mathcal{R} \cup S, E)$ with edges $E_\mathcal{R} = \{(C,s)\in\mathcal{R}\times S \mid C \circ s \}$.
\end{definition} 

Given an embedding of $G$ and a set $\mathcal{R}$ of faces of $G$, we can obtain an embedding of $\mathcal{D}_G$ by placing a node $v_M$ inside the face $R$ for each $R \in \mathcal{R}$.
This shows the  following observation.
\begin{observation}[\cite{GoemansW1998, BermanY2012}]
  If $\mathcal{R}$ is a set of faces of $G$, then the debit graph is planar. 
\end{observation}
  
Note that for $\mathcal{R}$ a set of cycles, a cycle $R \in \mathcal{R}$, the number of nodes $|R \cap S|$ that $R$ pays for in the hitting set is the degree of $R$ in the debit graph.

Recall the definition of the {\sc Subset FVS} problem, which seeks a minimum-weight node set $X$ which intersects all cycles from $\mathcal C_T$, the collection of cycles in $G$ which contain some node from a given set $T\subseteq V(G)$.
Observe that each node of $ S'_3$ has a witness cycle in $G^{S_{\bar{t}}}_3$; therefore, it is an inclusion-minimal hitting set for the collection $\mathcal C_T$ with $T =  S'_3$.
Goemans and Williamson \cite[Lemma 4.2]{GoemansW1998} 
showed that any inclusion-minimal hitting set for $\mathcal C_T$ has a laminar set of  witness cycles, which implies that there is a laminar set of pseudo-witness cycles $\mathcal{A}$ for hitting set $ S'_3$.
\begin{proposition}[{\cite[Lemma 4.2 specialized for {\sc Subset FVS}]{GoemansW1998}}]
\label{thm:minimalsubsetfvs}
  Let $G'$ be a planar graph and let $T \subseteq V(G')$.
  Let $\mathcal{C}_T$ be the set of cycles of $G'$ containing at least one node of $T$, and let $X$ be an inclusion-minimal hitting set for $\mathcal C_T$.
  Then there is a laminar set of cycles $\mathcal{A} = \{A_x \mid  x \in X \}$, satisfying $A_x \in \mathcal{C}_T$ and $A_x \cap X = \{x\}$.
\end{proposition}
Applying \autoref{thm:minimalsubsetfvs} to $G'=G_3$ and $X = T =  S'_3$ implies there is a laminar set $\mathcal{A} = \{A_x \mid  x \in S'_3 \}$ of cycles satisfying $A_x \cap S'_3 = \{x\}$.  
In other words, $\mathcal{A}$ is a laminar set of pseudo-witness cycles for $S'_3$.
Note that cycles of $\mathcal{A}$ may not be even, hence they may be pseudo-witness cycles for~$S'_3$, but not necessarily witness cycles for nodes of $ S'_3$.

Recall that, during the current iteration, our algorithm incremented
the blended inequalities of the cycles in a $2/3$-quasi-perfect tiling
$\tilT_H$ of $H$.  
Recall $H$ is an inclusion-minimal pocket
  of $G^{S_{\bar{t}}}_2$.  
  By abuse of notation, let $\tilT_H$ be
  the corresponding cycles of $G^{S_{\bar{t}}}_3$.  
  Let $\mathcal{D}$ 
  be the debit graph formed using $G^{S_{\bar{t}}}_3$, the cycle
  set~$\tilT_H$ and hitting set $ S'_3$.  
  
Obtain graph~$\mathcal{D}'$ from $\mathcal{D}$ by replacing each even cycle $C$ containing two faces with the two faces that compose it.
To be precise, construct $\mathcal{D}'$ by first taking all nodes of $S'_3$ and all faces of $H$ that lie inside some even cycle of $\tilT_H$ as the vertex set.
For each edge $(C,v) \in E(\mathcal{D})$, if the cycle $C$ consist of two faces $f_1,f_2$ add the edges $(f_1,v)$ and $(f_2,v)$ to $\mathcal{D}'$, otherwise add the edge $(C,v)$ to~$\mathcal{D}'$ (see \autoref{onethir}).
Delete isolated vertices from $\mathcal{D}'$.
\begin{figure}
  \centering
  \includegraphics[scale=0.8]{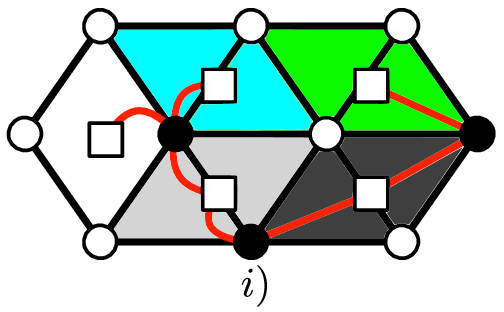}
  \hspace{1em}
  \includegraphics[scale=0.8]{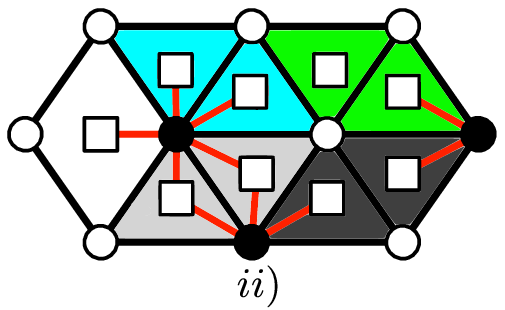} 
  \caption{Left: A possible debit graph $\mathcal{D}$ with the cycles of the tiling in \autoref{tiletomat}.
    Right: the graph~$\mathcal{D}'$ obtained by replacing each cycle with the faces that compose it.}
  \label{onethir}
\end{figure}

Let $\calT_{\textsf{Faces}(H)}$ be the ``face nodes'' of $\mathcal{D}'$.
Let $\mathcal{F}_{\textsf{all}(H)}$ denote the finite faces of $H$.
Let $\mathcal{F}_H $  denote the set of finite faces of $H$ that contain a hit node.
Observe that $M \cap  S'_3 = \emptyset$ for each $M \in  \mathcal{F}_{\textsf{all}(H)}  \backslash \mathcal{F}_H$.
Now
\begin{multline}
\label{M to M''}
  \sum_{M \in \tilT_H } |M \cap  S'_3|
  \leq \sum_{M \in \calT_{\textsf{Faces}(H)} } |M \cap  S'_3|\\
  \leq \sum_{M \in \mathcal{F}_{\textsf{all}(H)} } |M \cap  S'_3| - |\mathcal{F}_H \backslash \calT_{\textsf{Faces}(H)}|
     = \sum_{M \in \mathcal{F}} |M \cap  S'_3| - |\mathcal{F}_H \backslash \calT_{\textsf{Faces}(H)}| \enspace .
\end{multline} 
The first inequality holds, because for each cycle $C$ consisting of two faces $f_1$ and $f_2$ we have $|C\cap  S'_3| \leq |f_1 \cap  S'_3| + |f_2 \cap  S'_3|$.
The second inequality holds, because each face of $\mathcal{F}_H $  contains a hit node, and so $|C \cap  S'_3| \geq 1$ for each $C \in \mathcal{F}_H $.
The last inequality holds, because by definition $|M \cap  S'_3|=0$ for all $M \in \mathcal{F}_{\textsf{all}(H)}  \backslash \mathcal{F}_H $.

If our tiling covers $2/3$ of all finite faces, then $| \calT_{\textsf{Faces}(H)} | \leq 2|\tilT_H|$ and $(2/3)| \mathcal{F}_H | \leq |\calT_{\textsf{Faces}(H)}|$, so $|\mathcal{F}_H | \leq 3| \tilT_H| $.
Alas, one can show that a  tiling that covers $2/3$ of all finite faces does not always exist; see \autoref{tilecounterex}.
\begin{figure}
  \centering
  \includegraphics[scale=0.5]{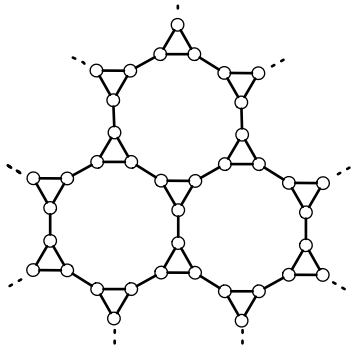}
  \caption{ A graph consisting of a tessellation of the plane with twice as many triangles as dodecagons. None of the triangles are adjacent, so a maximum tiling covers only the even dodecagons.\label{tilecounterex}}
\end{figure}
To overcome this impediment, we will show that $|\mathcal{F}_H | \leq 3| \tilT_H| $ holds for a $2/3$-quasi-perfect tiling. 
Suppose that our $2/3$-quasi-perfect tiling 
covers a $b$-fraction of the odd faces in~$\mathcal{F}_H $, and a $c$-fraction of the faces in~$\mathcal{F}_H$ which are even.
Let $\mathcal{F}_{\textsf{even}(H)}$ be the even finite faces of $\mathcal{F}_H $.
Then, as $\mathcal{F}_H \backslash \mathcal{F}_{\textsf{even}(H)} $ are the odd faces of~$\mathcal{F}_H$, and $\calT_{\textsf{Faces}(H)} \backslash \mathcal{F}_{\textsf{even}(H)}  $ are the odd faces covered by our tiling, it holds that $ b | \mathcal{F}_H \backslash \mathcal{F}_{\textsf{even}(H)}  | = | \calT_{\textsf{Faces}(H)} \backslash \mathcal{F}_{\textsf{even}(H)}  |$.
Simplifying, we get
\begin{equation*}
       b|\mathcal{F}_H | + (1-b)|\mathcal{F}_{\textsf{even}(H)} |
  \leq |\calT_{\textsf{Faces}(H)} |
  \leq 2| \tilT_H | -  | \mathcal{F}_{\textsf{even}(H)}  | \enspace .
\end{equation*}
By rearranging, we get $b | \mathcal{F}_H \backslash  \mathcal{F}_{\textsf{even}(H)}  | +2| \mathcal{F}_{\textsf{even}(H)} | \leq 2 | \tilT_H | $.
Noting that $b(1-c)+2c \geq 2/3$, and rearranging once more, yields
\begin{equation*}
       \frac{2}{3}|\mathcal{F}_H |
  \leq b|\mathcal{F}_H \backslash \mathcal{F}_{\textsf{even}(H)}| + 2|\mathcal{F}_{\textsf{even}(H)}|
  \leq |\calT_{\textsf{Faces}(H)}| \leq 2|\tilT_H| \enspace .
\end{equation*}
Noting that $|\mathcal{F}_{\textsf{even}(H)} | / |\mathcal{F}_H| = c$ and $b(1 - c) + 2c \geq 2/3$, we get
\begin{equation}
\label{bding size of M''}
  3|\tilT_H| \geq \frac{3}{2}(b(1-c)+2c)|\mathcal{F}_H|
             \geq |\mathcal{F}_H |  \enspace .
\end{equation} 
By \eqref{M to M''}, in order to bound $\sum_{M \in \tilT_H } |M \cap  S'_3|$, it suffices to bound $ \sum_{M \in \mathcal{F}} |M \cap  S'_3| $. To do this, we prove the following, which extends the work by Berman and Yaroslavtsev~\cite[Theorem 4.1]{BermanY2012}.
\begin{theorem}
\label{by2.4}
  Let $H$ be an inclusion-wise minimal pocket of $G$.
  Let $S \subset V(G)$ be a set of nodes with some set $\mathcal{A}$ of laminar pseudo-witness cycles.
  Let $\mathcal{R}$ be a set of finite faces of $H$ such that each cycle of $\mathcal{A}$ contains a face of $\mathcal{R}$ in its interior. 
  Then      $\sum_{M \in \mathcal{R} } |M \cap S| \leq \frac{18}{7}|\mathcal{R}|$.
 \end{theorem}
We defer the proof of \autoref{by2.4} to \autoref{by2.4Proof}.

Let $\mathcal{A}$ be a set of laminar witness cycles for $ S'_3$.
If we were to set $ \mathcal{R}= \mathcal{F}_H $  (the set of finite faces of $H$ incident to a hit node), then each cycle $A \in \mathcal{A}$  contains a face of $\mathcal{R}$ in its interior, namely any face inside $A$ that is incident to the hit node of $ S'_3$ on $A$.
Thus, $  S'_3, \mathcal{A}$ and $\mathcal{R}$ meet the conditions of \autoref{by2.4}. 
 
To recap, we wish to bound the primal increase rate $\sum_{M\in \tilT_H } \sum_{v \in S} a^M_v$, so we analyze the  expression $\sum_{M \in \tilT_H } |M \cap  S'_3 |$. Recall from \autoref{blendG3} that $\sum_{v \in S} a^M_v$ is at most one more than $|M \cap  S'_3 |$ and $\sum_{v \in S} a^M_v=|M \cap  S'_3 |$ if $M$ contains no twin edge.
We bound $\sum_{M \in \tilT_H } |M \cap  S'_3 |$ by looking at the quantity $\sum_{M \in \mathcal{F}_H } |M \cap  S'_3 |$, because~$\mathcal{F}_H $ fits the conditions of \autoref{by2.4}.
One could then use $|\mathcal{F}_H | \leq 3|\tilT_H | $ (by \eqref{bding size of M''}), to  bound $\sum_{M\in \tilT_H } \sum_{v \in S} a^M_v$ in terms of the dual increase rate~$|\tilT_H |$. 
We will use $3|\tilT_H| \geq \frac{3}{2}(b(1-c)+2c)|\mathcal{F}_H|$ to obtain a stronger bound.

Let $\calT$ be our $2/3$-quasi-perfect tiling from \autoref{2o3quasperexist}.
Recall from \autoref{quasiperdef} that the fraction~$\beta$ of odd finite faces that are covered by the tiling, and the fraction $\psi$ of finite faces of~$H$, that are even satisfy $\beta (1-\psi)+2\psi \geq \alpha$.
Let $\mathcal{A}$ be a set of pseudo-witness cycles in~$H$ for $ S'_3$, the corresponding set for the hitting set $S'$ returned by our algorithm.
Define $\mathcal{R}= \mathcal{F}_H $.
We have that every cycle of $\mathcal{A}$ contains a face of~$\mathcal{R}$ in its interior.
Thus, $\mathcal{R}, \mathcal{A}$ and~$ S'_3$ satisfy the conditions of \autoref{by2.4}.
Therefore,
\begin{align}
  \label{bound}
  \sum_{M \in \tilT_H } |M \cap  S'_3| \leq \left(\sum_{M\in \mathcal{F}_H }|M \cap  S'_3 |\right) - |\mathcal{F}_H \backslash \calT_{\textsf{Faces}(H)} |
                                        \leq \frac{18}{7}|\mathcal{F}_H | -|\mathcal{F}_H \backslash \calT_{\textsf{Faces}(H)} | \enspace .
    \end{align} 
Note that $ \sum_{v \in S} a^M_v \leq |M \cap S |$, unless $M$ contains a twin edge.
If $M \in \calT$ is the disjoint union of two odd faces which share an edge, then~$M$ will not contain a twin edge.
That is,~$M$ can only contain a twin edge if $M \in \mathcal{F}_{\textsf{even}(H)} $, so $M$ is an even face then.
So 
\begin{equation}
\label{tilingbd}
        \sum_{M\in \tilT_H } \sum_{v \in S} a^M_v
   \leq \sum_{M\in \tilT_H } |M \cap S | + |\mathcal{F}_{\textsf{even}(H)}  |
   \leq \frac{18}{7}|\mathcal{F}_H | -|\mathcal{F}_H \backslash \calT_{\textsf{Faces}(H)} | + |\mathcal{F}_{\textsf{even}(H)}  | \enspace .
\end{equation}
Recall that $c = |\mathcal{F}_{\textsf{even}(H)}  | / |\mathcal{F}_H |$ is the fraction of finite faces of $\mathcal{F}_H $  which are even, and that $b = | \calT_{\textsf{Faces}(H)}  \backslash \mathcal{F}_{\textsf{even}(H)}   | / |\mathcal{F}_H \backslash \mathcal{F}_{\textsf{even}(H)}  |$ is the fraction of odd finite faces of $\mathcal{F}_H $  covered by our tiling. 
Note that
\begin{align*}
  |\mathcal{F}\backslash \calT_{\textsf{Faces}(H)} | & = |\mathcal{F}_H \ \mathcal{F}_{\textsf{even}(H)}  | - |\calT_{\textsf{Faces}(H)} \backslash \mathcal{F}_{\textsf{even}(H)}   |\\
  & = |\mathcal{F}\backslash \mathcal{F}_{\textsf{even}(H)}  | - b|\mathcal{F}_H \backslash \mathcal{F}_{\textsf{even}(H)}   |
   = (1-b)(1-c)|\mathcal{F}_H | \enspace .
\end{align*}

We now recall \eqref{bding size of M''}, by which $3|\tilT_H | \geq \frac{3}{2}(b(1-c)+2c)| \mathcal{F}_H |$.

Substituting these bounds for $|\mathcal{F}_H | $ and $|\mathcal{F}_H \backslash \calT_{\textsf{Faces}(H)}  |$  into \eqref{tilingbd}, we obtain 
\begin{align*}
  \sum_{M\in \tilT_H }\sum_{v \in S}a^M_v & \leq c|\mathcal{F}_H | + \frac{18}{7}\left(\frac{2}{b(1-c)+2c}|\tilT_H |\right) - (1-b)(1-c) |\mathcal{F}_H |\\
                                             & = \frac{2c}{b(1-c)+2c}|\tilT_H | + \frac{18}{7}\left(\frac{2}{b(1-c)+2c}|\tilT_H |\right) - \frac{2(1-b)(1-c)}{b(1-c)+2c}|\tilT_H | \enspace .
\end{align*}
If we maximize the right-hand side factor
$\frac{2c}{(b(1-c)+2c)}+\frac{36}{7(b(1-c)+2c)} - \frac{2(1-b)(1-c)}
{(b(1-c)+2c)}$ subject to\linebreak $b(1-c)+2c \geq 2/3$, we obtain
that the right-hand side is bounded by $\frac{47}{7} |\tilT_H |$.
 
This completes the proof of \autoref{thm:apxmain} modulo the proof
of \autoref{2o3quasperexist}; i.e., the fact that large
quasi-perfect tilings can be computed efficiently. The remaining part
of this paper will provide details for this remaining task. 

\subsection{Obtaining a $2/3$-quasi-perfect tiling}
\label{2o3quasitil}
We now show how to find the $2/3$-quasi perfect tiling in line \ref{alg:qpt} of \autoref{tiling}.
The following result states that the minimal pockets picked by the algorithm have such tilings.

\thmcmd*

To prove this theorem we will use the following lemma.

\begin{lemma}
\label{nopseudo}
  For any set $S$, any pseudo-pocket contained in $G^S_2$  contains an even cycle. 
\end{lemma}
\begin{proof}
  Informally speaking, the proof will show that any pseudo-pocket without even cycles contains an odd cycle for which only two nodes have   outside neighbours; this, however, cannot appear in the 2-compression, as we would have replaced this cycle by an edge in $G^S_2$.

  Suppose, for sake of contradiction, that $G^S_2$ contained a pseudo-pocket $Q$ without even cycles.
  Since each node of $Q$ is in an even cycle of $G_2$ and $Q$ contains no even cycle,  $Q$ contains exactly two nodes $u$ and $v$ with neighbours outside $Q$, and each node of $Q$ lies on a $u$-$v$ path of $Q$.
  Let $B_u$ and~$B_v$ be the blocks of $Q$ containing $u$ and $v$ in the block graph $\mathcal{B}$ of $Q$, respectively (see \autoref{PseudNoEven}).

  If $\mathcal{B}$ was not a path, then there would be some block $B_1$ that does not lie on a $B_u$-$B_v$ path in~$\mathcal{B}$, and thus there would be a node of $B_1$ that would not lie on a $u$-$v$ path in~$Q$---a contradiction.
  Hence, $\mathcal{B}$ is a path.
  \begin{figure}[h]
      \centering
      \includegraphics[scale=0.8]{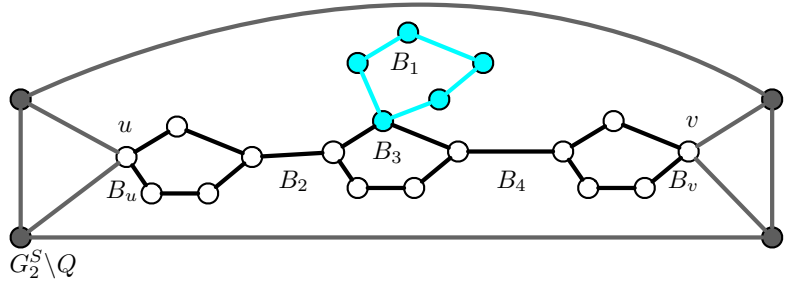}
      \caption{Graph $Q$ consisting of blocks labelled $B_1, B_2, B_3, B_4, B_u, B_v$. Block $B_1$ depicted in blue contains nodes not on any $u$-$v$ path, which is a contradiction. }
      \label{PseudNoEven}
  \end{figure}

  Let $B$ be a block of $Q$. 
Suppose, for sake of contradiction, that $B$ contains a cycle $C$ and a node~$v'$ of $C$ with a neighbour $u' \in V(B) $ outside $C$. 
Since $v'$ is not a cut node, there is a path~$P$ from $u'$ to $C \backslash v'$.
Construct the $u'$-$v'$ path $P'$ from $P$ by traversing $P$ from $u'$ to the first node~$w'$ of~$C \backslash v'$ and appending to that a $w'$-$v'$ path in $C$.
Since $Q$ contains no even cycles, the cycles $P' \cup {v'u'}$ and $C$ are odd.
Then the cycle formed by the edges  $ E(C) \Delta E(P' \cup {v'u'} ) $, that is edges of~$C$ or $P' \cup {v'u'}$, but not both, has length $ |E(C)| + |E(P' \cup {v'u'})|- 2|E(C) \cap E(P' \cup {v'u'}) | $ which is even, and hence a contradiction.
 Thus if $B$ contains a cycle then it does not contain nodes outside the cycle, or put simply $B$ is a cycle.
  Since we assume $B$ contains no even cycles, $B$ is an odd cycle.
  Thus, the blocks of $Q$ are odd cycles or edges. 
  Since~$Q$ contains at least one cycle, there is an odd cycle $C'$. 
Since $\mathcal{B}$ is a path, $C'$ contains 2 nodes  $a$ and $b$ with neighbours outside $C'$.  
 However,~$G^S_2$ cannot contain such an odd cycle, as that we would have contracted the two $a$-$b$ paths of $C'$ to parallel edges and then replaced them by a twin edge; see \autoref{cyclecompression}. This completes the proof.
  \begin{figure}[h] 
    \centering
    \includegraphics[scale=0.8]{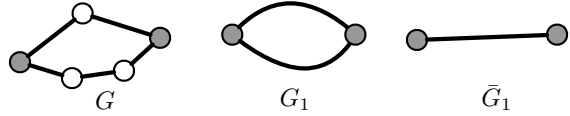}
    \caption{Cycle is replaced by an edge in 2-compression.}
    \label{cyclecompression}
  \end{figure}

  \end{proof}
For any set $S$, if $G^S_3$ contained a pseudo-pocket $Q$ without even cycles, then $Q$ was obtained from a subgraph~$Q'$ of $G^S_2$ by subdividing edges. 
Then $Q'$ would be a pseudo-pocket of $G^S_2$ without even cycles. 
This contradicts \autoref{nopseudo}. 
This shows the following corollary.
\begin{corollary}
For any set $S$, any pseudo-pocket of $G^S_3$ contains an even cycle.
\end{corollary}

Recall from \autoref{match-tile} and the  paragraph afterwards, that a  pseudo-tiling of our graph corresponds to the union of a matching of the dual graph and a set of even faces. A tiling corresponds to the union of a matching of the dual graph not containing any edge incident to the infinite face and a set of even finite faces.  
Under this correspondence, the existence of large pseudo-tilings is a much more natural thing to prove.
Let us first formally define a large pseudo-tiling.
\begin{definition}
\label{pseudo-perfect2}
  Let $\alpha \in (0,1)$.
  A pseudo-tiling $\mathcal{T} $ is \emph{$\alpha$-pseudo-perfect} if it covers all even faces (including the infinite face if it is even) and a $\beta$-fraction of the odd faces, and a $\psi$-fraction of the faces of $H$ are even, where  
  \begin{equation}
  \label{pseudo-perfecteq2}
    \beta (1-\psi)+2\psi \geq \alpha \enspace .
  \end{equation}
\end{definition}
We will first prove the existence of large pseudo-perfect pseudo-tilings.
We fix an embedding of~$H$.
For any multigraph $W$, let $\mathsf{oc}(W)$ be the number of odd components of $W$. 
Recall pseudo-tilings correspond to matchings. 
Our proof will use  Tutte's Theorem stated below, which informally speaking, says that the absence of a large matching implies the existence of a small set of vertices whose removal results in a graph with a large number of connected components of odd size.
\begin{theorem}[Tutte's Theorem]
\label{TutteThm}
  For any graph $G$, the number of nodes of $G$ which are not covered by a maximum size matching of $G$ is at most 
  \begin{equation}
  \label{Tutt}
     \mathsf{oc}(G \backslash X) - |X| \enspace .
  \end{equation}
  for some $X \subset V(G)$.  
  Further, if some node $v \in V(G)$ is covered by every maximum matching of~$G$, then \eqref{Tutt} holds for some $X \subset V(G)$ containing $v$.
\end{theorem}

The main idea  of why such large pseudo-perfect pseudo-tilings should exist is that by Tutte's Theorem, the absence of a large pseudo-tiling implies that for some set $X$ of nodes of the dual graph $H^*$, the set  
of odd components of $H^* \backslash X$ is large relative to $|X|$.

Construct a new graph $H^{1}$ as follows.
Start with the graph $H^*$ and add as many edges as possible between nodes of $X$  while preserving planarity and not creating any faces of length~two (see \autoref{H1ocnHoc}). 

\begin{figure}[h]
    \centering
    \includegraphics[scale=0.8]{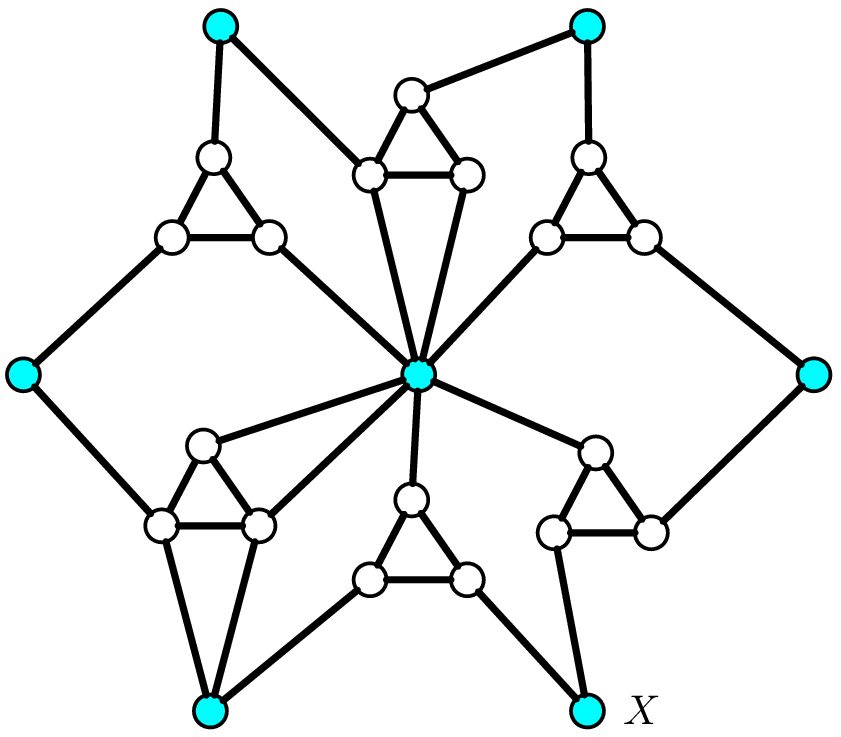} 
    \includegraphics[scale=0.8]{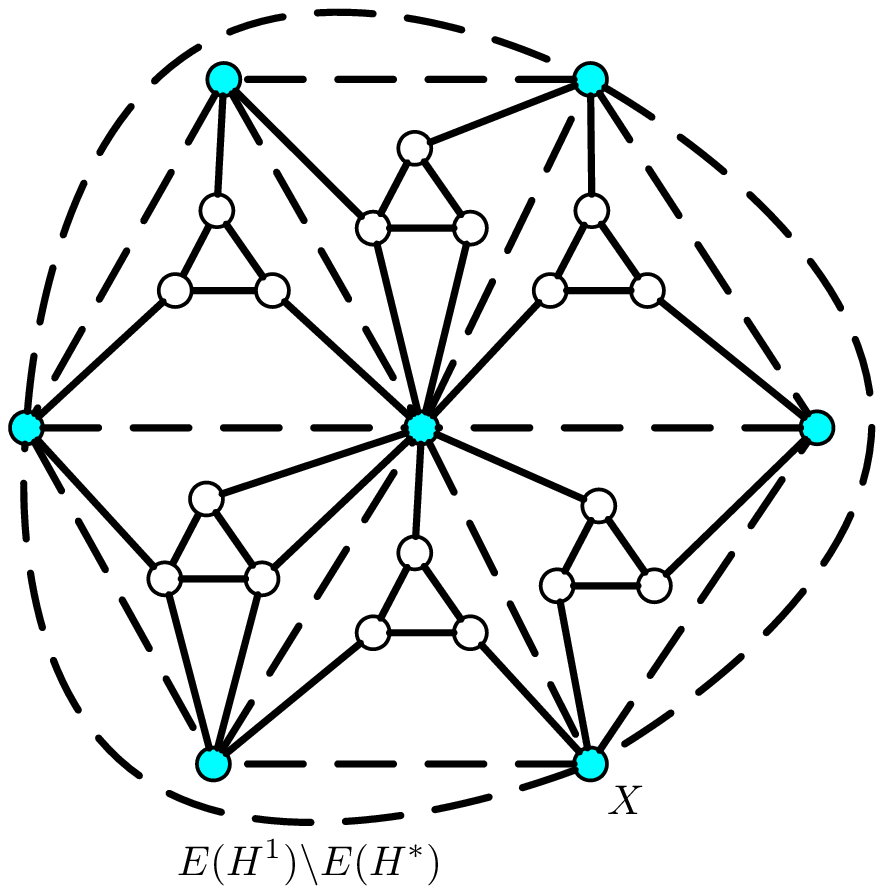}
    \caption{The graph $H^*$ with set $X \subset V(H^*)$ (depicted in blue) on the left. On the right, the graph $H^1$ obtained from $H^*$ by adding edges (dashed) between $X$. }
    \label{H1ocnHoc}
\end{figure}
We will show that each odd component of $H^{1} \backslash X$ lies in a different face of $H^{1}[X]$
and that $H^{1}$ contains at most two faces of length two. 
Thus using Euler's formula, $|E(H^{1}[X])| \leq 3 |V(H^{1}[X])| -4$, $H^{1}[X]$ does not have too many edges.
The crucial observation is that since each odd component of $H^{1} \backslash X$
lies in a different face of $H^{1}[X]$, each node $x \in X$  is adjacent to more other nodes of $X$ in $H^{1}$ than there are odd components of  $H^{1} \backslash X$ which contain a neighbour of~$x$. By facial region, we mean the region of the plane bounded by a face. 
We will also show there are at most two odd components $J_1, J_2$ for which at most two nodes of $X$ have neighbours in $J_i$, see \autoref{Deg2nddual} ii). There, for the odd component $J_i$, there are two nodes $u,w \in X$ which have neighbours in $J_i$. \autoref{Deg2nddual} iii) shows the ``corresponding dual graph'' $Q_i$ which contains only two nodes $s$ and $d$ with neighbours outside $Q_i$, which contradicts the fact that $H$ contains no pseudo-pockets. 
We can then show that the number of odd components is at most 
 $2/3$ the number of edges of $H^{1}[X]$ plus  $\frac{2}{3}$, which will contradict that the set of odd components is large.

\begin{lemma}
\label{pseudo-tiling}
  Let $H$ be as in \autoref{tiling}, that is, $H$ is a minimal pocket of $G^S_2$.
  Then $H$ has a $2/3$-pseudo-perfect pseudo-tiling.
\end{lemma}
\begin{proof}
  Suppose, for sake of contradiction, that $H$ does not have a 2/3-pseudo-perfect pseudo-tiling. 
  Recall that each edge of the dual graph $H^*$ of $H$ between two nodes which correspond to odd faces in $H$ corresponds to an even cycle of $H$.
  Thus, we may think of pseudo-tilings as the union of a set of even faces and a matching on the odd faces.
  Let $Y$ be the set of even faces of $H$.

  Consider a maximum matching of the odd faces of $H$, that is, a maximum matching $Q$ of $H^* \backslash Y$.
  Assume that $Q$ misses a $(1-b)$-fraction of the odd faces (of $H$), that is, $(1-b) = \frac{q'}{q}$, where $q'$ is the number of odd faces not incident to an edge of the matching, and $q$ is the total number of odd faces.  
  By \autoref{TutteThm} applied to $G = H^* \backslash Y$ (by an abuse of notation we also use $Y$ to denote the nodes of $H^*$ which correspond to faces of $Y$), there is a set of nodes of $ V(H^*) \backslash Y $ such that removing these nodes creates a relatively large number of odd components.
  More precisely, for some $X \subset V(H^*) \backslash Y$ we have
  \begin{equation}
  \label{tuttee2}
    (1-b) |V(H^* \backslash Y )| \leq \mathsf{oc}(H^* \backslash (X \cup Y))  -|X| \enspace .
  \end{equation}
  Tutte's Theorem also says that if $h_\infty$ is  
   covered by every maximum matching of $H^* \backslash Y$, then we may pick $X$ containing $h_\infty$.
  By rearranging \eqref{tuttee2}, we obtain $|V(H^*)| - |Y| - b|V(H^* \backslash Y )| \leq \mathsf{oc}(H^* \backslash (X \cup Y)) -|X|$. Subtracting $|Y|$ from both sides, we get  
  \begin{equation}
  \label{tuttee2p5}
     |V(H^*)| - 2|Y| - b|V(H^* \backslash Y) | \leq \mathsf{oc}(H^* \backslash (X \cup Y)) - |X \cup Y| \enspace .
  \end{equation}

  Note that a $|Y|/|V(H^*)|$-fraction of all the faces of $H$ are even, and by definition, a $b$-fraction of all the odd faces are covered by $Q$.
  There is a pseudo-tiling $\mathcal{T}$ corresponding to $Y \cup Q$ under \autoref{match-tile} and the paragraph afterwards.

  Let $J_1,\hdots,J_\ell$ be the odd components of $H^* \backslash (X \cup Y)$.
  Let $\hat{H}$ be the graph obtained from $H^*$ by contracting each $J_i$ deleting created parallel edges and loops. 
  For $i = 1,\hdots,\ell$ let $j_i$ be the node obtained by contracting $J_i$; let $J = \{j_1,\hdots,j_\ell\}$.
  Let $H'$ be an edge maximal (multi) graph obtained from $\hat{H}$ by adding edges between nodes of $X \cup Y$ while preserving planarity and not creating any faces of length two.

We will show the following 3 claims.
 \begin{claim}
  \label{sumodeg}
    The inequality  $\sum_{i=1}^\ell |\delta_{\hat{H}}(j_i) | \geq 3\ell-2$ holds. 
  \end{claim}
   \begin{claim}
  \label{EdgesOfH'}
    It holds $|E(H'[X \cup Y])| \leq 3| X \cup Y| - 3$.
  \end{claim}
   \begin{claim}
  \label{JvsH'}
    It holds $|E(H') \cap J \times (X \cup Y)| \leq 2|E(H'(X \cup Y)|$.
  \end{claim} 
We defer the proofs for now and show how to finish the proof given these claims.
From \autoref{sumodeg}, it follows that $3|J| -2  \leq \sum_{i=1}^\ell |\delta_{H'}(j_i)| = |E(H') \cap J \times (X \cup Y)|$. Thus, by \autoref{JvsH'} and \autoref{EdgesOfH'}, it follows that
  \begin{equation*}
    3|J| - 2 \leq 2|E(H'(X \cup Y)|
           \leq 6|X \cup Y| - 6 \enspace .
  \end{equation*}
  So $| X \cup Y| \geq 0.5 |J| =\ell$. 

Suppose for a contradiction that the pseudo-tiling $\mathcal{T}$ is not 2/3-pseudo-perfect, 
  then \eqref{pseudo-perfecteq2} of \autoref{pseudo-perfect2} is violated, that is, 
  \begin{equation*}
    b(1- (|Y|/|V(H^*)|) ) + 2|Y|/|V(H^*)|  < 2/3 \enspace .
  \end{equation*}
  After simplifying, we obtain $2|Y| +b |V(H^* \backslash Y) | < \frac{2}{3}|V(H^*)|$.
  Therefore, it holds\linebreak $\frac{1}{3}|V(H^*)| <  |V(H^*)| - 2|Y| - b|V(H^* \backslash Y)|$.
  Substituting this into the left-hand side of \eqref{tuttee2p5}, we obtain
  \begin{equation}
  \label{tuttee3}
     \frac{1}{3}|V(H^*)|    < |V(H^*)| - 2|Y| - b |V(H^* \backslash Y) |
                         \leq \mathsf{oc}(H^* \backslash X \cup Y)  -|X \cup Y|.
  \end{equation}

  From $|J|+|X \cup Y| \leq |V(H^*)|$ and $|X \cup Y| \geq \frac{1}{2}|J|$, we get $\frac{2}{3}|V(H^*)| \geq |J|$. 
  Consequently,
  \begin{equation*}
    \frac{1}{3}|V(H^*)| \geq \frac{1}{2}|J|
                        \geq |J| - |X \cup Y|
                           = \mathsf{oc}(H^* \backslash (X \cup Y) ) - |X \cup Y|,
  \end{equation*}
  which contradicts \eqref{tuttee3}.
  Therefore, $\mathcal{T}$ is $2/3$-pseudo-perfect.
  This completes the proof of the lemma.
  \end{proof}
 We use the notation in the proof of \autoref{pseudo-tiling} throughout the rest of this section.
 Denote by~$Q_i$ the subgraph of $H$ induced by the faces of $H$ corresponding to $J_i$. 
  Given a node $v \in V(H^*)$, denote by $v^* \subset H$ the face of $H$ which  $v$ corresponds to.
  Let $h^*_\infty$ denote the infinite face of $H$ and $h_\infty$ the node of the dual graph $H^*$ corresponding to $h^*_\infty$.
  
  We need the following remark for the next claim.
\begin{remark}\label{QiCycle}
 If $h\infty \notin J_i$, then the infinite face of $Q_i$ is a cycle.
\end{remark}
\begin{proof}
    Assume for a contradiction the infinite face $f_{Q_i \infty}$ of $Q_i$ was not a cycle. 
    Then there is a cycle $C$ of $f_{Q_i \infty}$  for which the region bounded by $C$ contains at least one and not all finite faces of~$Q_i$. Let $ F $ be the set of finite faces of $Q_i$ bounded by $C$. Since $C$ ``separates" the faces of $F$ from the other finite faces of $Q_i$, the vertices of $J_i$ corresponding to faces of $F$ are not reachable from the other vertices of $J_i$ in $H^* \backslash h_\infty$. 
\end{proof}
  
  We argue that $Q_i$ cannot be a pseudo-pocket.

  If $Q_i$ is a pocket, then since $Q_i$ is contained in $H$, this contradicts the fact that $H$ is an inclusion-wise minimal pocket. 
  Otherwise, $Q_i$ is a pseudo-pocket with no even cycle, which by \autoref{nopseudo}, cannot appear in the 2-compression of a graph.
  The following claim shows that a certain condition on $j_i$ implies $Q_i$ is a pseudo-pocket, which implies that such a condition cannot hold for $j_i$.
   \begin{claim}
  \label{claim:pseudopockets}
    Suppose the degree $|\delta_{\hat{H}}(j_i)|$ of $j_i$ in $\hat{H}$ is at most $ 2$, $h\infty \notin J_i$ and no node of $Q_i$ on the infinite face has a neighbour outside $H$ (see node $t$ in \autoref{Deg2nddual} $vi)$).
    Then $Q_i$  is a pseudo-pocket. 
 \end{claim}
 We illustrate the previous claim in \autoref{Deg2nddual} $i)$-$iii)$. In $i)$, $j_i$ has two neighbours $u$ and $w$. In $iii)$, $Q_i$ is bounded by the two faces $u^*$ and $w^*$ and only the nodes $s$ and $d$ in $Q_i$, the two nodes of~$Q_i$ which belong to both $u^*$ and $w^*$, have neighbours outside $Q_i$.
 \begin{proof}
  Intuitively, the neighbours of $J_i$ in $ H^* \backslash J_i$ correspond to the faces of $H$ bound $J_i$. 
  Informally, if $J_i$ has only 2 neighbours $u,w$ in $ H^* \backslash J_i$ and $u,w \neq h_\infty$,  then the corresponding faces $u^*$ and $  w^*$ bound $Q_i$, which implies $Q_i$ is a pocket (see \autoref{Deg2nddual} $iii)$). 

  To be precise, suppose that $j_i$ has degree two and $u,w$ are the only nodes of $V(H^*) \backslash J_i$ with neighbours in $J_i$ (see \autoref{Deg2nddual} $ii)$). 
  Each edge $e$ on the infinite face $W_i$ of $Q_i$ lies on a face $a^*$ of $H$ where~$a$ is a node of $ H^* \backslash J_i$.
  The only nodes of $V(H^*)$ that have neighbours in $J_i$ are $u,w$. 
  Thus,~$a$ is either $u$ or $w$. 
  So the edge $e$ lies on one of the faces $u^*$ or $w^*$.
  We may assume $u \neq h_ \infty$.
  Recall that~$H$ contains no pseudo-pockets. 
  Therefore, the intersection of any two finite faces of a subgraph of $H$ with a common edge is a path.  
 Let $W_i$ denote the outside face of $Q_i$, which by \autoref{QiCycle} is a cycle. 
It follows that $A_1 = W_i \cap u^* $ is a path. 
Let $s$ and $d$ denote the endpoints of $A_1$. 
Since each edge of $A_1$ lies on a face of $Q_i$ and $u^*$, it does not lie on the face $w^*$.
 So $A_2 = W_i \cap w^*$ consists of the subgraph of $W_i$ formed by the nodes not in the interior of $A_1$. 
  Hence, $A_2$ is a path with endpoints $s$ and $d$.
  Thus, in the graph $H$, only nodes $s$ and $d$ of $Q_i$ can have neighbours in~$H \backslash Q_i$. 
  Thus, if no node of $Q_i$ has a neighbour outside $H$, then $Q_i$ is a pseudo-pocket of~$H$.

Now suppose $j_i$ has a single neighbour $u$.
Let $W_i$ denote the outside face of $Q_i$, which is a cycle.
If $u^*$ is the infinite face, then $W_i \cap u^*$ is the infinite face of $Q_i$, which is a cycle. 
In this case $Q_i = H$. 
Suppose $u \neq h_\infty$. 
Since each edge lies on two faces, each edge of $W_i$ lies on $u^*$. 
Note that faces of graphs are enclosed by closed walks such that each cycle  contains at most one node with a neighbour in the walk but outside this cycle.
Thus, there is exactly a single node $s \in W_i$ for which $s$ contains a neighbour in $u^* \backslash W_i$. 
This node $s$ is the only node of $Q_i$ with a neighbour outside $Q_i$, see \autoref{Deg2nddual} $vii)$.  
Thus, $Q_i$ is a pseudo-pocket of $H$.
  \begin{figure}[h]
    \centering
    \includegraphics[scale=0.8]{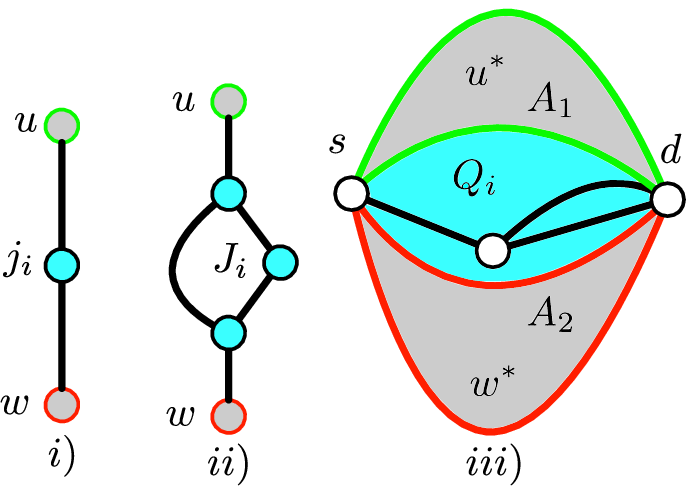} 
    \hspace{1mm}
    \includegraphics[scale=0.8]{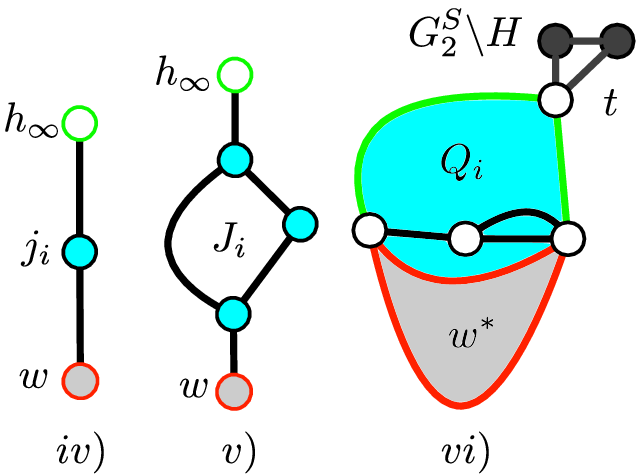}
     \includegraphics[scale=0.8]{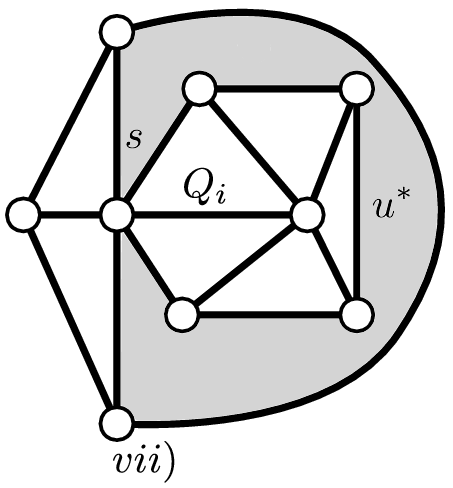}
     \caption{Figures $i),ii),$ and $iii)$ show how a degree two node in $\hat{H}$, not incident to $h_\infty$, which is shown in (i), corresponds to a pseudo-pocket, which is shown in (iii). 
     Figures $iv),v),vi)$ show the exception when conditions of \autoref{claim:pseudopockets} are not satisfied, that is, the node $j_i$ is adjacent to $h_\infty$, and a node $t$ on the infinite face of $Q_i$ has a neighbour outside $H$. In this case, $j_i$ may not correspond to a pseudo-pocket of $G^S_2$.
     The shaded nodes in vi) are part of $G^S_2 \backslash H$. Figure $vii)$ shows $Q_i$ bounded by a single face $u^*$. In this case $Q_i$ is also a pseudo-pocket.
     }
     \label{Deg2nddual}
  \end{figure}
    This completes the proof of \autoref{claim:pseudopockets}.
  \end{proof}

 The proof of \autoref{sumodeg} will use the fact that $H^* \backslash h_\infty$ is connected, which we prove next. 
   \begin{remark}\label{Hdual2conn}
   For the minimal pocket $H$ found by \autoref{tiling}, $H^* \backslash h_\infty$ is connected.
 \end{remark}
 \begin{proof}
 We show $H$ is 2-connected.
   Note that if $H$ has a cut node $v$, then some component of $H \backslash v$, say $H^1$, contains at most one node with a neighbour outside $H$.
  As a consequence, $H^1$ would be a smaller pocket, which would contradict the fact that $H$ is a minimal pocket. 
  Thus, $H$ is 2-connected.
  It is well known that if $H$ is two connected, then the infinite face $h^*_\infty$ is a cycle. 
  Thus each face of $H$ lies  in the finite region bounded by $h^*_\infty$ and thus $H^* \backslash h_\infty$ is connected.
     \end{proof}

  We now prove \autoref{sumodeg}, \autoref{EdgesOfH'} and \autoref{JvsH'}.
  \begin{proof}[Proof of \autoref{sumodeg}]
   
    We distinguish two cases.
    \begin{enumerate}
        \item \label{1J_i} Some $j_i$ contains only one neighbour in $X$.
        \item \label{2J_i} Each $j_i$ contains at least 2 neighbours in $X$.
    \end{enumerate}
  
In Case \ref{1J_i}, we claim that for $j_a$ such that $a \neq i$, $|\delta_{\hat{H}}(j_a) | \geq 3$. We consider three sub-cases.
  
    Case \ref{1J_i}a) $h_\infty \notin J_i$ and the one neighbour that $J_i$ has in $X$ is not $h_\infty$.
    Then by \autoref{claim:pseudopockets} the subgraph of $H$ corresponding to the faces $J_i$ is a pocket, which contradicts our assumption that $H$ is a minimal pocket.
  
   Case \ref{1J_i}b) $h_\infty \notin J_i$ and the one neighbour that $J_i$ has in $X$ is $h_\infty$.
    Then $h_\infty$ separates~$J_i$ from the rest of $H^* \backslash h_ \infty$.
    That is, $J_i$ is a component of $ H^* \backslash h_\infty $. By \autoref{Hdual2conn}, $ H^* \backslash h_\infty $ is connected, so $J_i = H^* \backslash h_\infty$. Thus there do not exist $J_a$ for $a \neq i$ and the condition is trivially true.

    Case \ref{1J_i}c) $h_\infty \in J_i$.
    Then $J_i$ contains all  nodes that have neighbours outside $H$ and no other $J_a$ contains a node with a neighbour outside $H$. Thus for each $a \neq i$, $j_a$ satisfies $|\delta_{\hat{H}}(j_a) | \geq 3$.
  
    In all three sub-cases, $J_a$ does not contains a node with a neighbour outside $H$.
    So, $|\delta_{\hat{H}}(j_a) | \geq 3$ for all $a \in \{1,\hdots,\ell\} \backslash \{i\}$.
    
    Therefore in Case \ref{1J_i},  $\sum_{t=1}^\ell |\delta_{\hat{H}}(j_t) | \geq 3\ell-2$.
    This completes the analysis of Case \ref{1J_i}.
  
    In the Case \ref{2J_i}, each $J_i$ contains at least two neighbours in $X$.
    If $Q_i$  contains a node $v_i$ with a neighbour outside $H$ in the interior of the shared path between $Q_i$ and the infinite face of~$H$, then $v_i$ has degree~two in $H$.
    Thus, $v_i$ is incident to only faces $h_\infty$ and $J_i$.
    So $v_i$ does not lie in any~$Q_t$ for $t \neq i$.
    Since at most two nodes of $H$ have neighbours outside $H$, there are at most two $Q_a$ 
    that contain a node $v_a$ with a neighbour outside $H$ in the interior of the shared path between $Q_a$ and the infinite face of $H$.
    For these $Q_a$, $|\delta_{\hat{H}}(j_a)| \geq 2$. 
    For every other $Q_r$, $|\delta_{\hat{H}}(j_r)|$ is at least 3, and thus $\sum_{t=1}^\ell |\delta_{\hat{H}}(j_t) | \geq 3\ell-2$.  
  
    In either case, we get $\sum_{t=1}^\ell |\delta_{\hat{H}}(j_t) | \geq 3\ell-2$, as desired.    
    This completes the proof of \autoref{sumodeg}.
\end{proof}

  \begin{proof}[Proof of \autoref{EdgesOfH'}]
    First note that if $H'[X \cup Y]$ contains  parallel edges $e_1,e_2$ between two nodes $u,w \in X \cup Y$, then in the planar embedding of $H'$, there are nodes of $J$ that lie in the region bounded by $e_1$ and~$e_2$.  
    The faces corresponding to $u$ and $w$ in $H$ then bound a pocket unless one of those faces is the infinite face, and the region bounded contains a node with a neighbour outside $H$. 
    Hence, $H'[X \cup Y]$ contains at most two faces of length two. 
    Thus, if $ |X \cup Y| \geq 2 $, then $H'[X \cup Y]$  contains at most two more edges than a planar graph on at least two nodes, that is, at most $2 + 3|X \cup Y| - 5 = 3|X \cup Y| - 3$ edges. 
    Otherwise, $|X \cup Y | \leq 1  $, so $ |E(H'(X \cup Y)| = 0 $, which is at most  $ 3 | X \cup Y | -3 $. 
    This completes the proof of \autoref{EdgesOfH'}.
  \end{proof}

  \begin{proof}[Proof of \autoref{JvsH'}]
    We claim that in any embedding of $H'$ each node $r \in X \cup Y$ does not have two consecutive neighbours in $J$ in the clockwise orientation about $r$.  
    Assume that some $r \in X \cup Y$ has two consecutive neighbours $j_a,j_b \in J$.
    Consider the face containing the nodes $r,j_a,j_b$. 
    Let~$r'$ be a neighbour of $j_b$ in this face. 
    Then the edge $r r'$ can be added to $H'$ without creating a face of length~two, which contradicts the fact that $H'$ is an edge maximal multigraph with respect to planarity and not having faces of length two, that is, no edge can be added to $H'$ while maintaining planarity and not creating any face of length two. 

    This implies that, for each $x \in X \cup Y$, it holds
    \begin{equation*}
      |E(H') \cap J \times \{x \}|  \leq  |E(H') \cap (X \cup Y) \times \{x \}| \enspace .
    \end{equation*}

    Summing up over all each $x \in X \cup Y$ we obtain  
    \begin{align*}
      |E(H') \cap J \times (X \cup Y)|  &    = \sum_{x \in X \cup Y} |E(H') \cap J \times \{x \}|\\
                                        & \leq \sum_{x \in  X \cup Y} |E(H') \cap (X \cup Y) \times \{x \}|\\
                                        & \leq 2|E(H'(X \cup Y)| \enspace .
    \end{align*}
    Thus, it holds $|E(H') \cap J \times (X \cup Y)| \leq 2|E(H'(X \cup Y)|$.

    This completes the proof of \autoref{JvsH'}.
  \end{proof}

So let $\mathcal T$ be a $2/3$-pseudo-perfect pseudo-tiling of $H$.
Let $\beta'$ be the fraction of odd faces of $H$ which are covered by $\mathcal T$, and let $\psi'$ be the fraction of even faces of $H$.
Next, we will show that if $\mathcal{T}$ covers more faces than a maximum tiling of $H$, then $\mathcal{T}$ satisfies a slightly stronger condition than $2/3$-pseudo-perfect, namely, $\beta'(1-\psi')|V(H^*)| + 2\psi'|V(H^*)| \geq \frac{2}{3}|V(H^*)|+\frac{4}{3}$. 
Formally, this means:
\begin{lemma}
\label{tilingimproved}
  Let $H$ be as in \autoref{tiling}, that is, $H$ is a minimal pocket of $G^S_2$. 
  Suppose that any maximum size pseudo-tiling of $H$ covers the infinite face.
  Then $H$ has a pseudo-tiling covering a $\beta'$-fraction of all odd faces such that
  \begin{equation}
    \label{eqn:lemma3old}
    \beta'(1-\psi')|V(H^*)| + 2\psi'|V(H^*)| \geq \frac{2}{3}|V(H^*)| + \frac{4}{3} \enspace .
  \end{equation}
\end{lemma}
\begin{proof}

  To show the statement of \autoref{tilingimproved}, we will need the following slight strengthening of \autoref{EdgesOfH'}.
  \begin{claim}
  \label{quasitileproof}
    Suppose that any maximum size pseudo-tiling of $H$ covers the infinite face and $H$ admits no $2/3$-quasi-perfect tiling.
    Then $|E(H'(X \cup Y)| \leq 3| X \cup Y| - 4$.
  \end{claim}
  \begin{proof}[Proof of \autoref{quasitileproof}.]
  Let $\mathcal{T}$, $X$, $Y$ be as in the proof of \autoref{pseudo-tiling}. 
  If the infinite face of $H$ is odd, then by assumption, every maximum matching of $H^* \backslash Y$ covers $h_\infty$.
  Recall this meant we picked~$X$ to contain $h_\infty$.
  Otherwise, $h_\infty \in Y$.
  So we may assume $ h_\infty \in X \cup Y$.
  
 By \autoref{Hdual2conn}, if $ X \cup Y = \{ h_\infty \}$, then $ \mathsf{oc}(H^* \backslash (X \cup Y)) =1$, which means that either $\beta' = 1$  or $X = \emptyset$. 

 Let us first prove the claim in the case that $ X \cup Y = \{ h_\infty \}$. 
 Suppose $ X \cup Y = \{ h_\infty \}$.

  In case $\beta'=1$, then a maximum pseudo-tiling covers all odd faces, and a maximum tiling covers all but at most one odd face. 

  In case $X= \emptyset$, we get that at most one odd face is not covered by a maximum pseudo-tiling.
  As $X \cup Y = \{h_\infty\}$, the infinite face is even.
  Thus, at most one odd face is missed by a maximum tiling.

  In either case, a maximum tiling $\mathcal{T}$ misses at most one odd face. 

  Let $\beta$ be the fraction of odd finite faces that are covered by $\mathcal{T}$, and $\psi$ the fraction of finite faces of~$H$ that are even.
  As $\mathcal{T}$ misses at most one odd face, it holds\linebreak $\beta (1- \psi )(|V(H^*)|-1) \geq (1- \psi )(|V(H^*)|-1) -1$.
  
  First, assume that $H$ contains some even finite face.
  Then
  \begin{align*}
    \beta (1- \psi )(|V(H^*)|-1) + 2 \psi (|V(H^*)|-1) & \geq (1- \psi )(|V(H^*)|-1) -1 + 2 \psi (|V(H^*)|-1)\\
                                                       &   = (|V(H^*)|-1) -1 + \psi (|V(H^*)|-1)  \geq (|V(H^*)|-1) \enspace .
  \end{align*}
  So, $\mathcal{T}$ is $2/3$-quasi-perfect.
  
  Second, suppose that $H$ contains no even finite faces.
  If $H$ contains a single odd finite face, then it contains no even cycle, which is a contradiction.
  If $H$ contains exactly two odd finite faces, then since the maximum tiling misses at most one odd finite face, all odd finite faces of $H$ are covered; so, a maximum tiling is 1-quasi-perfect.
  
  If $H$ contains three or more finite faces.
  Then noting that at most one face of $H$ is not covered by  $\mathcal{T}$, it follows that $\psi (|V(H^*)|-1) + \beta (1- \psi )(|V(H^*)|-1) \geq |V(H^*)|-2 $.
  So the inequality $ \beta (1- \psi )(|V(H^*)|-1) + 2 \psi (|V(H^*)| -1) \geq  \psi (|V(H^*)|-1) + \beta (1- \psi )(|V(H^*)|-1) \geq |V(H^*)|-2  $ holds, which by algebra yields  $\beta (1- \psi ) + 2 \psi \geq \frac{|V(H^*)|-1}{|V(H^*)|-2} $.
  As $|V(H^*)|-1 \geq 3$, $ \frac{|V(H^*)|-1}{|V(H^*)|-2} \geq \frac{2}{3}$, so the tiling is $\frac{2}{3}$-quasi-perfect. 
 
  Henceforth, we assume $ |X \cup Y| \geq 2$.

  Suppose first that $|X \cup Y| =2$.
  If $H'(X \cup Y)$ contains three parallel edges $e_1,e_2,e_3$, then it contains three faces of length two each bounded by a pair of parallel edges.
  Since $H'$ contains no parallel edges, the set $R_i$ of nodes lying in the face bounded by the parallel edges $e_ie_{i+1}$ where $e_4=e_1$ is nonempty for $i=1,2,3$.
  For illustration, see \autoref{3paredge}(i).
  Let $R^*_i$ be the subgraph of $H$ induced by the nodes that lie on a face which is the dual of a node of $R_i$. 
  \begin{figure}[h]
    \centering
    \includegraphics[scale=0.8]{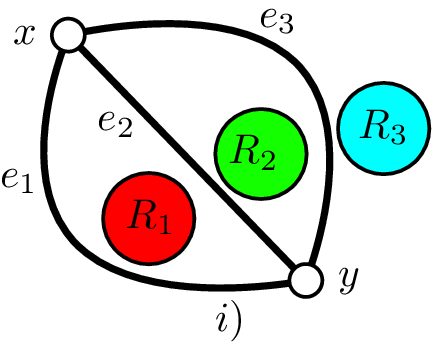}   \hspace{3mm}
    \includegraphics[scale=0.8]{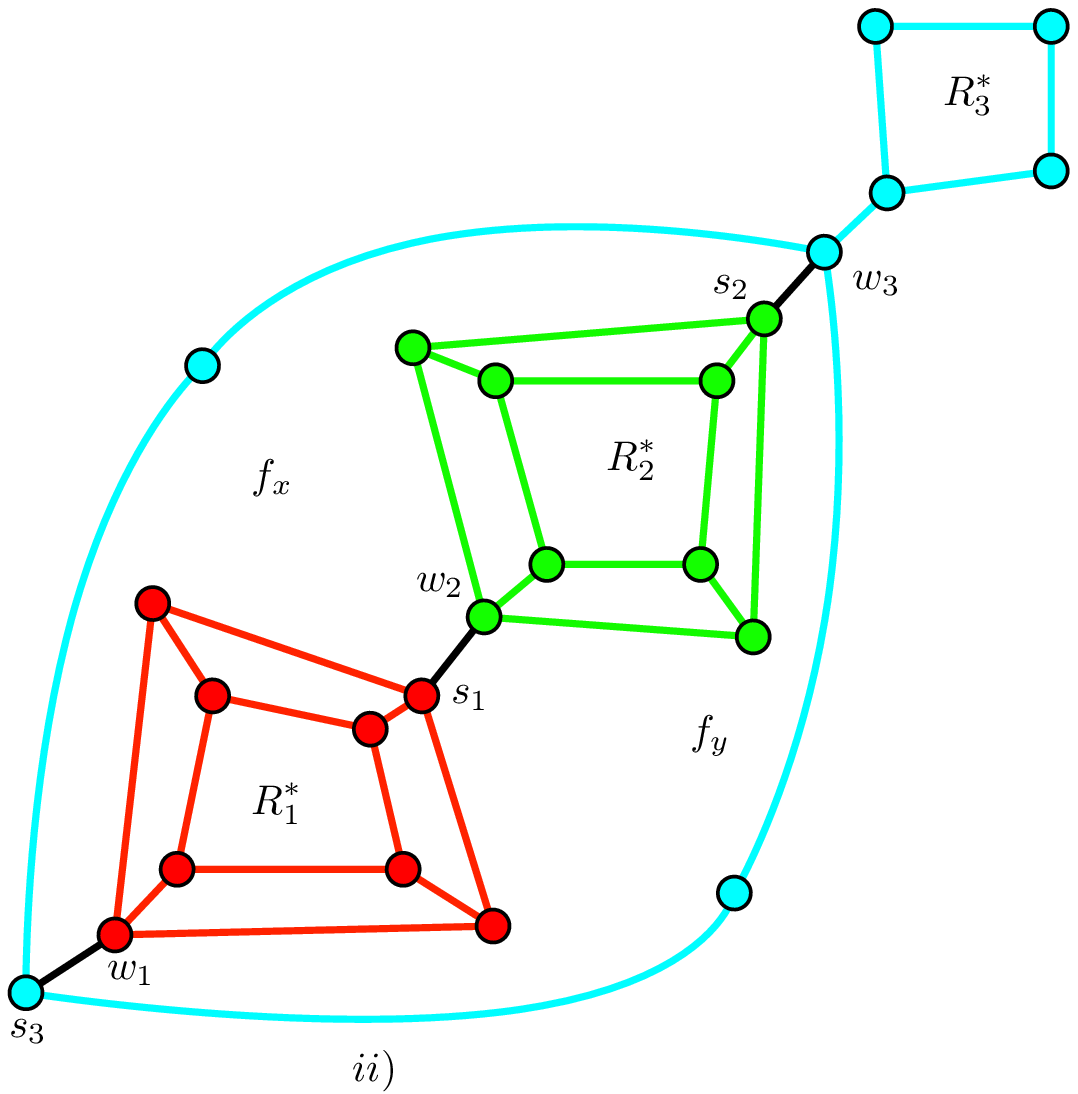}
    \caption{In i) 3 parallel edges $e_1,e_2,e_3$ bounding node sets $ R_1,R_2,R_3 $ in $H^*$ and in ii) the duals $R^*_1,R^*_2,R^*_3$ in $H$ respectively. }
    \label{3paredge}
  \end{figure}
  Note that $R^*_i$ lies in a region~$T_i$ bounded by the faces $f_x$ and $f_y$ of $H$ which are dual to $x$ and $y$ respectively, see \autoref{3paredge}ii).
  Denote by $ w_i $ and $s_i$ the two nodes on the boundary of the region $T_i$ belonging to both faces $f_x$ and~$f_y$.
  Then $R^*_i$ is a pocket unless some node of $V(R^*_i) \backslash \{ w_i, s_i \} $ has a neighbour outside $H$. 
  Note that for $i \neq j$, $ V(R^*_i) \backslash \{ s_i,w_i \}  )  \cap (V(R^*_j) \backslash \{ s_j,w_j \} ) = \emptyset $.
  Since at most 2 nodes of $H$ have neighbours outside $H$, at least one $R_i$ has no node with a neighbour outside $H$ and thus is a pocket, which is a contradiction.  
  Hence $H'(X \cup Y)$ contains only two edges, and $ |E(H'(X \cup Y))| \leq 2 = 3|X \cup Y| -4$.  

  Second, suppose that $|X \cup Y| >2$.
  By Euler's formula, any planar graph with nodes $X \cup Y$ without faces of length two has at most $3|X \cup Y|-6$ edges.
  Suppose $F_i$, $i=1,\hdots, p$  are faces of length~two in $H'(X \cup Y)$, let $q_i,r_i$ be the nodes, and $e_i,d_i$ the edges of $F_i$.
  Since $H'$ contains no parallel edges, the subgraph $R_i$ of $H'$ lying inside the region bounded by $F_i$, is nonempty.
  Let $R^*_i$ denote the subgraph of $H$ induced by the set of nodes that lie on a face which is the dual of a node of $R_i$.
  Then each $R^*_i$ lies in a region $T_i$ bounded by two faces $f_{q_i}$ and $f_{r_i}$  which are the dual of $q_i$ and~$r_i$.
  Let $ s_i, w_i $ be the nodes of $H$ on the boundary of $T_i$ that belong to both faces $f_{q_i}$ and $f_{r_i}$.
  See \autoref{OnePar} for an illustration.

  If no node of $V(R^*_i) \backslash \{ s_i,w_i \}$ has a neighbour outside $H$, then $R^*_i$ is a pseudo-pocket.
  Note that for $i \neq j$, the sets $ V(R^*_i) \backslash \{s_i,w_i\}$ and $V(R^*_j) \backslash \{ s_j,w_j \}$ are disjoint.
  Hence, if there were three length-2 faces $F_1,F_2,F_3$, then one of $R_1,R_2,R_3$ would be a pseudo-pocket, which is a contradiction.
  Thus, $H'(X \cup Y)$ contains at most two faces of length two.
  Therefore, there are two edges $e'_1e'_2$ that we can remove from $H'(X \cup Y)$ such that $H'(X \cup Y) \backslash \{e'_1,e'_2\}$ contains no face of length two.
  Hence, $|E(H'(X \cup Y) \backslash \{e'_1, e'_2\}| \leq 3|X \cup Y| - 6$ and $|E(H'(X \cup Y))| \leq 3|X \cup Y| - 4$.

This completes the proof that $|E(H'(X \cup Y))| \leq 3|X \cup Y| - 4$.
  \begin{figure}[h]
    \centering
    \includegraphics[scale=0.8]{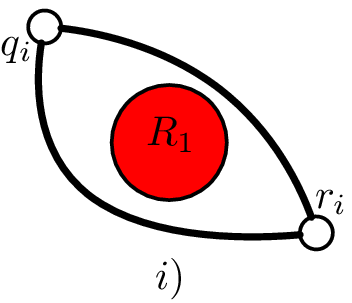}
    \includegraphics[scale=0.8]{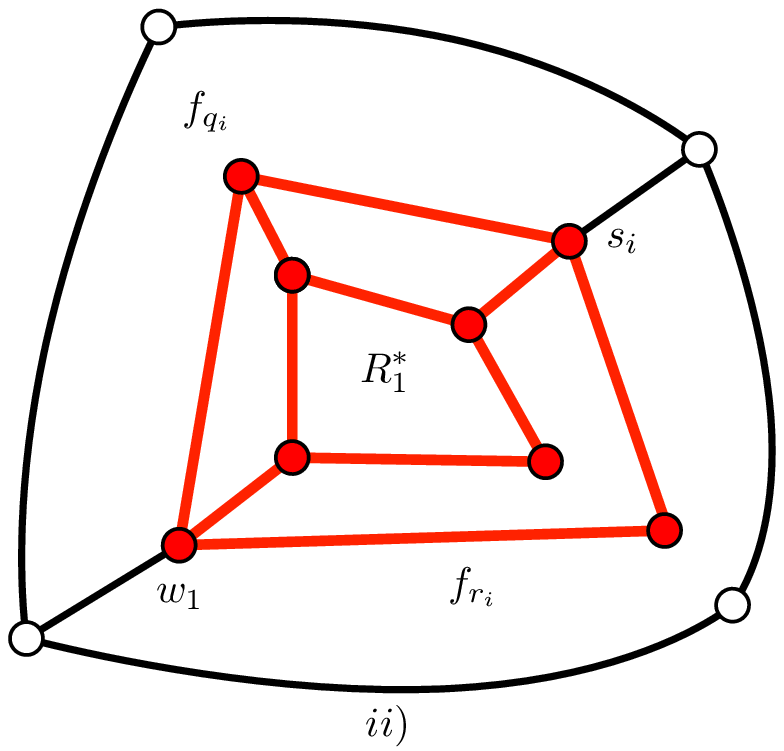}
    \caption{ One the left, one parallel edge in $H'$ bounding a region containing a set of nodes $R_1$. On the right is shown the dual graph, in which $R_1^*$ is a pocket. 
    }
    \label{OnePar}
  \end{figure}

  \end{proof}
  
  By assumption, $\mathcal{T}$ covers more faces than a maximum tiling of $H$.
  Suppose, for sake of contradiction, that $\beta'(1-c)|V(H^*)| + 2c|V(H^*)| < \frac{2}{3}|V(H^*)|+4/3$.
  Then, by \autoref{quasitileproof}, it holds $|E(H'(X \cup Y)| \leq 3| X \cup Y| - 4$.
  Further, by \autoref{JvsH'}, it holds $|E(H') \cap J \times \{x \}| \leq 2|E(H'(X \cup Y)|$.
  Also, by \autoref{sumodeg}, we have $\sum_{i=1}^\ell |\delta_{\hat{H}}(j_i) | \geq 3\ell-2$. 
  So in summary, we obtain
  \begin{equation*}
    3|J| -2  \leq \sum_{i=1}^\ell |\delta_{H'}(j_i)| = |E(H') \cap J \times (X \cup Y)| \enspace .
  \end{equation*}
  Hence, $3|J| -2 \leq 2|E(H'(X \cup Y)| \leq 6|X \cup Y | -8$.  
  Therefore,
  \begin{equation}
  \label{quasiJvsXcupY}
    |J| \leq 2|X \cup Y| - 2 \enspace .
  \end{equation}
  Substituting $\psi'=\frac{|Y|}{|V(H^*)|}$ into $\beta'(1-\psi')|V(H^*)| + 2c|V(H^*)| < \frac{2}{3}|V(H^*)|+4/3$, we obtain\linebreak $\beta'(1- (|Y|/|V(H^*)|) )|V(H^*)| + 2|Y| < \frac{2}{3}|V(H^*)| + \frac{4}{3}$.
  So $2|Y| + \beta'|V(H^* \backslash Y)| < \frac{2}{3}|V(H^*)| + \frac{4}{3}$, and thus $\frac{1}{3}|V(H^*)| - \frac{4}{3} < |V(H^*)| - 2|Y| - \beta'|V(H^* \backslash Y)|$.
  Substituting this into the left-hand side of \eqref{tuttee2p5}, we obtain 
  \begin{equation}
  \label{tutte3}
    \frac{1}{3}|V(H^*)| - \frac{4}{3}     <  |V(H^*)| - 2|Y| - \beta'|V(H^* \backslash Y)|
                                       \leq \mathsf{oc}(H^* \backslash X \cup Y) - |X \cup Y| \enspace .
  \end{equation}
  Multiplying both sides by $(-1)$ and adding to $\mathsf{oc}(H^* \backslash X \cup Y)  + |X \cup Y| \leq |V(H^*)|$, we obtain $2|X \cup Y | < \frac{2}{3}|V(H^*)| + \frac{4}{3}$. Simplifying, we obtain $ |X \cup Y | < \frac{1}{3}|V(H^*)| + \frac{2}{3}$. 
  Thus,  
  \begin{equation}
  \label{tutte4}
    \mathsf{oc}(H^* \backslash X \cup Y) > 2|X \cup Y| - 2 \enspace .
  \end{equation}
  This, however, contradicts \eqref{quasiJvsXcupY}.
  Hence, $\beta'(1-\psi')|V(H^*)| + 2\psi'|V(H^*)| \geq \frac{2}{3}|V(H^*)|+4/3$, which completes the proof.
\end{proof}

\begin{theorem}
  Let $H$ be an inclusion-minimal pocket of $G^S_2$.
  Then we can obtain $2/3$-quasi-perfect tiling of~$H$ in polynomial time.
\end{theorem}
\begin{proof}
  We first show that $H$ admits a $2/3$-quasi-perfect tiling.
   Let us show that if some tiling $\mathcal{T}$ is $2/3$-pseudo-perfect, then it is $2/3$-quasi-perfect.
  Let $\beta'$ be the fraction of odd faces of $H$ that are covered by $\mathcal{T}$ and $\psi'$ the fraction of faces of $H$, that are even.
  As $\mathcal{T}$ is $2/3$-pseudo-perfect, it covers all even faces. 
  Since $\mathcal{T}$ is a tiling, the infinite face is odd.
  As the number of even finite faces is $\psi'|V(H^*)|$, so $\frac{\psi'|V(H^*)|}{|V(H^*)|-1}$ is the fraction of finite faces of $H$ that are even.
  $(1-\psi')|V(H^*)|$ is the number of odd faces of $H$, so $\beta'(1-\psi')|V(H^*)|$ is the number of odd faces of $H$ covered by $\mathcal{T}$. 
  Since the infinite face is odd, $(1-\psi')|V(H^*)|-1$ is the number of odd finite faces.  
  Thus $\frac{\beta'(1-\psi')|V(H^*)|}{(1-\psi')|V(H^*)|-1}$ is the fraction of odd finite faces of $H$ covered by $\mathcal{T}$. 
  Since 
  \begin{align*}
    &\qquad   \frac{\beta'(1-\psi')|V(H^*)|}{(1-\psi')|V(H^*)|-1} \left(1-\frac{\psi'|V(H^*)|}{|V(H^*)|-1}\right) +\frac{2\psi'|V(H^*)|}{|V(H^*)|-1}\\
    & = \frac{\beta'(1-\psi')|V(H^*)|}{(1-\psi')|V(H^*)|-1}(1-\psi') +  2\psi' + \left(2-\frac{\beta'(1-\psi')|V(H^*)|}{(1-\psi')|V(H^*)|-1}\right )\left(\psi'- \frac{\psi'|V(H^*)|}{|V(H^*)|-1}\right)\\ 
    &\leq \frac{\beta'(1-\psi')|V(H^*)|}{(1-\psi')|V(H^*)|-1}(1-\psi') +  2\psi' \\
    &\leq \frac{2}{3},
  \end{align*}
  it holds that $\mathcal{T}$ is $2/3$-quasi-perfect.

  If there is a maximum size pseudo-tiling that is also a tiling,  then it follows from \autoref{pseudo-tiling} that such a tiling is $2/3$-quasi-perfect.
  
  Otherwise, if no pseudo-tiling exists, the largest pseudo-tiling is larger than the largest tiling. 
  Let $\mathcal{T}$ be a maximum size pseudo-tiling.
 
  If the infinite face of $\mathcal{T}$ is even, consider the tiling~$\mathcal{T'}$ obtained by removing the infinite face from $\mathcal{T}$.
  Let $\psi^{(1)}:=(\psi'|V(H^*)|-1)/(|V(H^*)|-1)$ be the fraction of finite faces of $H$ which are even.
  As the infinite face is even, $\beta'$ is the fraction of odd finite faces of $H$ which are covered by~$\mathcal{T}'$.
  It holds that
  \begin{align*}
    \beta'(1-\psi^{(1)})(|V(H^*)|-1) + 2\psi^{(1)}(|V(H^*)|-1) &    = \beta'|V(H^*)|(1-\psi')+\psi'|V(H^*)| - 1\\
                                          & \geq \frac{2}{3}|V(H^*)| + \frac{4}{3} - 1\\
                                          &    = \frac{2}{3}(|V(H^*)| - 1) \enspace .
  \end{align*}
  So $\mathcal{T}'$ is $2/3$-quasi-perfect.

  If the infinite face is odd, consider the tiling $\mathcal{T'}$ obtained by removing the even cycle covering the infinite face from $\mathcal{T}$.
  Let $\psi^{(2)}:=\psi'|V(H^*)|/(|V(H^*)|-1)$ be the fraction of finite faces of $H$ that are even.
  At least $\beta'|V(H^*)|-2$ of the finite faces of $H$ are covered by $\mathcal{T'}$ 
  so the fraction $\beta'''$ of finite odd faces of $H$ that are covered satisfies $b'' \geq (\beta'|V(H^*)|-1)/(1-\psi^{(2)})(|V(H^*)|-1)$.
  Therefore,
  \begin{align*}
    b''(1-\psi^{(2)})(|V(H^*)|-1) +2\psi^{(2)}(|V(H^*)|-1) & \geq (\beta'|V(H^*)|-1) + 2c|V(H^*)|\\
                                             & \geq \frac{2}{3}|V(H^*)| + \frac{4}{3} - 1\\
                                             &    = \frac{2}{3}(|V(H^*)| - 1) \enspace .
  \end{align*}
  Hence also in this case, $\mathcal{T}'$ is $2/3$-quasi-perfect. 

  Finally, since a tiling corresponds to the union of a matching and a set of even faces,  finding a maximum tiling of $H$ corresponds to finding a maximum matching of the odd finite faces of $H$.
  Computing such a maximum matching can be done in polynomial time.
\end{proof}

\subsection{Proof of \autoref{by2.4}}\label{by2.4Proof}
In this section we will prove \autoref{by2.4}.

Let $G,H, \mathcal{R}, S, \mathcal{A}$ be as in the statement of \autoref{by2.4}.
Recall the notion of debit graph of $G$ from \autoref{debit}.
Let $\mathcal{D}_G$ be the debit graph of $G$ with respect to $S$.

We introduce the notion of ``balance'', which captures for subsets $\mathcal{R'} \subseteq \mathcal{R}$ of cycles are incident to more or less than $18/7$ nodes of $S$ 
on average.

\begin{definition}
\label{balance}
  For each subset $\mathcal{R}' \subseteq \mathcal{R}$, its \emph{balance} $\mathsf{bal}(\mathcal{R}')$ is the quantity $|\mathcal{R}'| - \frac{7}{18} |E_\mathcal{R}'|$.
\end{definition} 
Our proof follows the same methodology as Berman and Yaroslavtsev~\cite{BermanY2012}.
First, it shows a pseudo-witness cycle that is not a face and is minimally so, that is any pseudo-witness cycle lying in the finite region bounded by it is a face,  has balance at least $1-\frac{7}{18}$.
Then it uses this to apply a reduction on $G$.
We will use the following result of theirs.
\begin{proposition}[{\cite[Lemma 4.3]{BermanY2012}}]
\label{complex}
   Let $W$ be a planar graph, 
   $\hat{S}$ be a set of nodes of $W$ and $Q \subset \hat{S} $ be a set of nodes of $W$ that we call \emph{outer nodes}.
  Let $\mathcal{R}_W$ be a set of faces of $W$  such that each non-outer node of $\hat{S} \cap W $  has a pseudo-witness cycle in $\mathcal{R}_W$.  
  If $W$ contains $a \leq 2 $ outer nodes, 
  then  $\mathsf{bal}( \mathcal{R}_W $) $ \geq  1 - \frac{7}{18} a $.
\end{proposition}

\begin{definition}
  If all nodes of a pseudo-witness cycle $A$ are contained in $H$, call $A$ a \emph{hierarchical} pseudo-witness cycle. 
  Otherwise, call $A$ a \emph{crossing} pseudo-witness cycle. 
  Denote the set of crossing pseudo-witness cycles by $\hat{\mathcal{A}}$.
\end{definition}
We are now ready to complete the proof of \autoref{by2.4}.
We begin by reductions on our instance $(G,H,\mathcal{R}, \mathcal{A},S)$ which simplify our instance and do not increase the balance. If after applying this reduction our instance has positive balance, then our instance had positive balance before the reduction. We define the reduction below.
\begin{definition}
\label{reduction}
  We define the following reduction on our instance $(G,H,\mathcal{R}, \mathcal{A},S)$.
  If $H$ contains a hierarchical pseudo-witness cycle $A$ that is not a face of $\mathcal{R}$, delete all nodes, edges and faces of~$\mathcal{R}$ inside $A$ from $H$ and add $A$ to $\mathcal{R}$.
  If $H$ does not contain a hierarchical witness cycle, we call the instance $(G,H,\mathcal{R}, \mathcal{A},S)$ \emph{reduced}.
\end{definition}
Let $\mathcal{R}_C$ be the faces in $\mathcal{R}$ contained in the region bounded by $C$.
Let $H^1,\mathcal{R}^1 $ be the result of applying the reduction in \autoref{reduction} on $H,\mathcal{R}$.
The balance of $H^1, \mathcal{R}^1$ is equal to
\begin{align*}
    |(\mathcal{R} \backslash \mathcal{R}_C)\cup\{C\}| - \sum_{M\in(\mathcal{R} \backslash \mathcal{R}_C)\cup\{C\}}|M \cap S|
  & = |\mathcal{R}| - \sum_{M \in \mathcal{R}}|M\cap S| - (|\mathcal{R}_C| + 1 - (\sum_{M \in \mathcal{R}_C}|M \cap S|) + 1)\\
  & = \mathsf{bal}(H)+1 - \mathsf{bal}(\mathcal{R}_C)-\frac{7}{18}.
\end{align*}
That is to say, the reduction changes the balance by $1 - \mathsf{bal}(\mathcal{R}_C)-\frac{7}{18}$, which by \autoref{complex} is non-positive. 

Thus if after applying the reduction in \autoref{reduction}, our instance has positive balance then it initially had positive balance. 
We know apply the reduction in \autoref{reduction} until our instance is reduced, for simplicity we will continue to call this graph $H$.

The crossing pseudo-witness cycles $\hat{\mathcal{A}}$  partition $H$ into \emph{regions}, see \autoref{PartitionOfH}.
\begin{figure}
  \centering
  \includegraphics[scale=0.8]{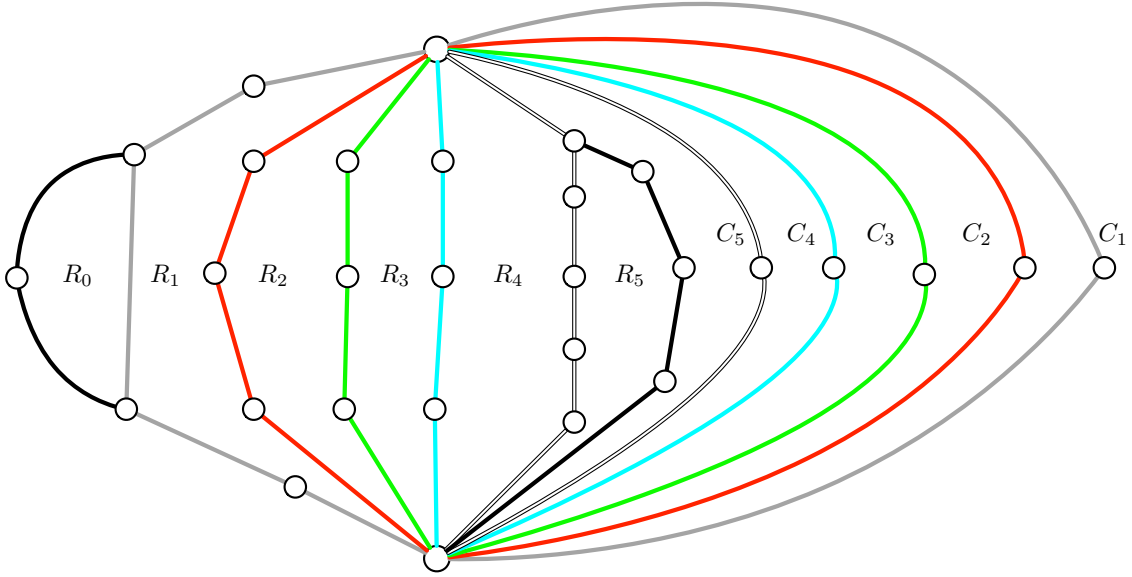}
  \caption{Pseudo-witness cycles $C_1,C_2,.., C_5$ divide $H$ into regions $R_0,R_1,\hdots,R_5$.\label{PartitionOfH}}
\end{figure}
That is, consider the subgraph $K \subset H $ consisting of nodes and edges lying on a witness cycle of $\hat{\mathcal{A}}$ or on the outside face of $H$. 
The regions are defined as the portions of the plane bounded by the finite faces of $K$.  
Define a \emph{subpocket} \cite{BermanY2012} as the subgraph of $H$ consisting of the nodes and edges lying in or on the boundary of a region.

\begin{proposition} [\cite{BermanY2012}]
\label{partitiontosubpocket}
  The regions that the set of crossing cycles $\hat{\mathcal{A}}$ partition the plane into satisfy the following. 
  For each region, there is a set $\tilde{\mathcal{A}}$ of at most 2 pseudo-witness cycles of $\hat{\mathcal{A}}$ such that each node  bounding the region either does not lie on a pseudo-witness cycle in $\hat{\mathcal{A}}$  or lies on a cycle of $\tilde{\mathcal{A}}$. 
\end{proposition}

By the reduction described in \autoref{reduction} 
each non-crossing cycle of $\mathcal{A}$ is a face. 
Since by \autoref{partitiontosubpocket}, the outside face of each subpocket $W$ contains nodes from at most two crossing pseudo-witness cycles, and contains all nodes that belong to pseudo-witness cycles lie on the outside face, there are at most two hit nodes of $W$ whose pseudo-witness is not a face and they must lie on the outside face of $W$.
Hence, each subpocket satisfies the conditions of \autoref{complex} and hence has positive balance.
Thus, $H$ has positive balance, that is, $0 \leq |\mathcal{R}| - \frac{7}{18} |E_\mathcal{R}| =|R| - \sum_{M \in R} |M \cap S| $.
Rearranging, $\sum_{M \in R} |M \cap S| \leq \frac{18}{7}|\mathcal{R}|$, which completes the proof of \autoref{by2.4}.
$\qed$

\bibliographystyle{abbrvnat}
\bibliography{EvenCycleTransversal}

\end{document}